\documentclass[11pt]{article}
\usepackage[utf8x]{inputenc}
\usepackage[english]{babel}
\usepackage{graphicx}   
\usepackage[pdftex, colorlinks=true, linkcolor=black, citecolor=black]{hyperref}
\usepackage{braket}
\usepackage{amsmath}
\usepackage{amssymb}
\usepackage{amsthm}
\usepackage{mathtools}
\usepackage{bbm}		
\usepackage{array}		
\usepackage{color}
\usepackage[usenames,dvipsnames]{xcolor}
\usepackage{caption}
\usepackage{enumitem}	
\usepackage{wrapfig}

\allowdisplaybreaks

\theoremstyle{definition}\newtheorem{definition}{Definition}[section]

\theoremstyle{plain}\newtheorem{thm}{Theorem}
\newtheorem{lem}{Lemma}[section]
\newtheorem{cor}[lem]{Corollary}
\newtheorem{prop}[lem]{Proposition}

\theoremstyle{remark}

\newcommand{\lemit}[1]{\begin{enumerate}[label={(\alph*)}, ref={\thelem\alph*}]{#1}\end{enumerate}}	

\renewcommand{\hat}[1]{\widehat{#1}}											
\renewcommand{\tilde}[1]{\widetilde{#1}}										

\newcommand{\ls}{\lesssim}													
\newcommand{\gs}{\gtrsim}													

\newcommand{\lr}[1]{\left\langle #1 \right\rangle} 							
\newcommand{\norm}[1]{\lVert#1\rVert}   										
\newcommand{\onorm}[1]{\lVert#1\rVert_\mathrm{op}}							

\newcommand{\R}{\mathbb{R}}													
\newcommand{\N}{\mathbb{N}}													
\newcommand{\fC}{\mathfrak{C}}
\newcommand{\cN}{\mathcal{N}}

\newcommand\mydots{,\makebox[1em][c]{.\hfil.\hfil.},}							
\newcommand\mycdots{\makebox[1em][c]{$\cdot$\hfil$\cdot$\hfil$\cdot$}}		

\newcommand{\Tr}{\mathrm{Tr}}												
\renewcommand{\d}{\mathop{}\!\mathrm{d}}										

\renewcommand{\i}{\mathrm{i}}												
\newcommand{\e}{\mathrm{e}}

\let\textl\l
\renewcommand{\l}{\ell}

\newcommand{\D}{d}

\newcommand{\Hb}{H^\beta}
\newcommand{\Vext}{{V^\mathrm{ext}}}
\newcommand{\vb}{v^\beta}
\newcommand{\vbot}{\vb_{12}}
\newcommand{\vboth}{\vb_{13}}
\newcommand{\vbij}{\vb_{ij}}

\newcommand{\vbarpt}{\overline{v}^\pt}

\newcommand{\pt}{{\varphi(t)}}
\newcommand{\ps}{{\varphi(s)}}
\newcommand{\pr}{{\varphi(r)}}
\newcommand{\pz}{{\varphi_0}}

\newcommand{\Hbog}{\mathbb{H}_\mathrm{Bog}}

\newcommand{\pp}{p^\varphi}
\newcommand{\qp}{q^\varphi}
\newcommand{\Pp}{P^\varphi}

\newcommand{\qpz}{q^\pz}
\newcommand{\Ppz}{P^\pz}
\newcommand{\ppt}{p^\pt}
\newcommand{\qpt}{q^\pt}
\newcommand{\Ppt}{P^\pt}

\newcommand{\hpt}{h^{\pt}}

\newcommand{\mpt}{\mu^\pt}

\newcommand{\Hpt}{\tilde{H}^\pt}

\newcommand{\Ubog}{\tilde{U}_\varphi}
\newcommand{\psibog}{\tilde{\psi}_\varphi}

\newcommand{\Tex}{T^\mathrm{ex}_{d,v,\Vext}}

\newcommand{\chik}{\chi^{(k)}}
\newcommand{\xipk}{\xi^{(k)}_\varphi}

\newcommand{\xipz}{\xi^{(0)}_\varphi}
\newcommand{\xipo}{\xi^{(1)}_\varphi}
\newcommand{\xipN}{\xi^{(N)}_\varphi}
\newcommand{\xip}{\xi_\varphi}
\newcommand{\xipt}{\xi_\pt}
\newcommand{\xiptk}{\xi^{(k)}_\pt}

\newcommand{\xipZ}{\xi_{\varphi_0}}
\newcommand{\xipzk}{\xi^{(k)}_{\pz}}

\newcommand{\xipzN}{\xi^{(N)}_{\varphi_0}}
\newcommand{\xitpt}{\tilde{\xi}_\pt}
\newcommand{\xitptk}{\tilde{\xi}_\pt^{\,(k)}}

\newcommand{\Np}{\mathcal{N}_\varphi}
\newcommand{\Npz}{\mathcal{N}_{\varphi_0}}
\newcommand{\Npt}{\mathcal{N}_{\varphi(t)}}

\newcommand{\UNp}{\mathfrak{U}_N^\varphi}
\newcommand{\UNpz}{\mathfrak{U}_N^{\pz}}
\newcommand{\UNpt}{\mathfrak{U}_N^{\pt}}

\newcommand{\FNp}{{\mathcal{F}^{\leq N}_{\perp \varphi}}}
\newcommand{\FNpz}{{\mathcal{F}^{\leq N}_{\perp\pz}}}
\newcommand{\FNpt}{{\mathcal{F}^{\leq N}_{\perp\pt}}}
\newcommand{\Fpt}{{\mathcal{F}_{\perp\pt}}}

\newcommand{\Cpt}{\mathcal{C}^\pt}
\newcommand{\Cps}{\mathcal{C}^\ps}
\newcommand{\Cpso}{\mathcal{C}^{\varphi(s_1)}}
\newcommand{\Cpst}{\mathcal{C}^{\varphi(s_2)}}
\newcommand{\Cpr}{\mathcal{C}^\pr}

\newcommand{\Qpt}{\mathcal{Q}^\pt}
\newcommand{\Qps}{\mathcal{Q}^\ps}

\newcommand{\Qpst}{\mathcal{Q}^{\varphi(s_2)}}
\newcommand{\Qpr}{\mathcal{Q}^\pr}

\newcommand{\psia}{\psi^{(a)}_\varphi}
\newcommand{\psiao}{\psi^{(a+1)}_\varphi}

\newcommand{\psio}{\psi^{(1)}_\varphi}
\newcommand{\psit}{\psi^{(2)}_\varphi}
\newcommand{\psith}{\psi^{(3)}_\varphi}

\newcommand{\Tnk}{T_n^{(k)}}
\newcommand{\Tak}{T_a^{(k)}}
\newcommand{\Tkk}{T_k^{(k)}}
\newcommand{\Taa}{T_a^{(a)}}
\newcommand{\Tzz}{T_0^{(0)}}

\newcommand{\Ttnk}{\tilde{T}_n^{(k)}}
\newcommand{\Ttak}{\tilde{T}_a^{(k)}}
\newcommand{\Ttkk}{\tilde{T}_k^{(k)}}
\newcommand{\Ttaa}{\tilde{T}_a^{(a)}}
\newcommand{\Ttaoao}{\tilde{T}_{a+1}^{(a+1)}}

\newcommand{\Ttoo}{\tilde{T}_1^{(1)}}
\newcommand{\Tttt}{\tilde{T}_2^{(2)}}
\newcommand{\Ttot}{\tilde{T}_1^{(2)}}

\newcommand{\wl}{w_\lambda}
\newcommand{\lj}{l^j}
\newcommand{\llj}{l_\lambda^j}

\newcommand{\fcaz}{\fC_{\,a}}

\newcommand{\fcbz}{\fC_{\,b}}
\newcommand{\fczz}{\fC_{\,0}}

\newcommand{\cat}{C_{a}^{\,t}}

\newcommand{\cbt}{C_{b}^{\,t}}
\newcommand{\cbts}{C_{b}^{\,t,s}}

\newcommand{\cjts}{C_{j}^{\,t,s}}

\DeclareMathOperator*{\supp}{\mathrm{supp}}

\makeatletter
\DeclareFontFamily{OMX}{MnSymbolE}{}
\DeclareSymbolFont{MnLargeSymbols}{OMX}{MnSymbolE}{m}{n}
\SetSymbolFont{MnLargeSymbols}{bold}{OMX}{MnSymbolE}{b}{n}
\DeclareFontShape{OMX}{MnSymbolE}{m}{n}{
    <-6>  MnSymbolE5
   <6-7>  MnSymbolE6
   <7-8>  MnSymbolE7
   <8-9>  MnSymbolE8
   <9-10> MnSymbolE9
  <10-12> MnSymbolE10
  <12->   MnSymbolE12
}{}
\DeclareFontShape{OMX}{MnSymbolE}{b}{n}{
    <-6>  MnSymbolE-Bold5
   <6-7>  MnSymbolE-Bold6
   <7-8>  MnSymbolE-Bold7
   <8-9>  MnSymbolE-Bold8
   <9-10> MnSymbolE-Bold9
  <10-12> MnSymbolE-Bold10
  <12->   MnSymbolE-Bold12
}{}

\let\llangle\@undefined
\let\rrangle\@undefined
\DeclareMathDelimiter{\llangle}{\mathopen}%
                     {MnLargeSymbols}{'164}{MnLargeSymbols}{'164}
\DeclareMathDelimiter{\rrangle}{\mathclose}%
                     {MnLargeSymbols}{'171}{MnLargeSymbols}{'171}
\makeatother

\title{Higher order corrections to the mean-field description of the dynamics of interacting bosons}

\author{Lea Boßmann\thanks{Fachbereich Mathematik, Eberhard Karls Universität Tübingen, 
	Auf der Morgenstelle 10, 72076 Tübingen, Germany; 
	and Institute of Science and Technology Austria, Am Campus 1, 3400 Klosterneuburg, Austria. \texttt{lea.bossmann@ist.ac.at}},\;
	Nata\v{s}a Pavlovi\'c\thanks{Department of Mathematics, The University of Texas at Austin, 2515 Speedway Stop C1200, Austin, TX 78712, USA. \texttt{natasa@math.utexas.edu}},\;
	Peter Pickl\thanks{Mathematisches Institut, Ludwig-Maximilians-Universität München, Theresienstr.\ 39, 80333 München, Germany; and Duke Kunshan University, No.\ 8 Duke Avenue, Kunshan, Jiangsu Province, 215316, China. \texttt{pickl@math.lmu.de}},\;
	and Avy Soffer\thanks{Department of Mathematics, Rutgers University, 110 Frelinghuysen Road, Piscataway, NJ 08854, USA. \texttt{soffer@math.rugers.edu}}
}

\date{\today}

\begin{document}
\maketitle

\begin{abstract}
\noindent
In this paper, we introduce a novel method for deriving higher order corrections to the mean-field description of the dynamics of interacting bosons. More precisely,
we consider the dynamics of $N$ $\D$-dimensional bosons for large $N$.
The bosons initially form a  Bose--Einstein condensate and interact with each other via a pair potential of the form $(N-1)^{-1}N^{\D\beta}v(N^\beta\cdot)$ for $\beta\in[0,\frac{1}{4\D})$.
We derive a sequence of $N$-body functions which approximate the true many-body dynamics in $L^2(\R^{\D N})$-norm to arbitrary precision in powers of $N^{-1}$.
The approximating functions are constructed as Duhamel expansions of finite order in terms of the first quantised analogue of a Bogoliubov time evolution.
\end{abstract}


\section{Introduction}
We consider a system of $N$ bosons in $\R^\D$, $\D\geq 1$, interacting with each other via pair interactions in the mean field scaling regime. 
The Hamiltonian of the system is given by
\begin{equation}\label{Hb}
\Hb(t):=\sum\limits_{j=1}^N\left(-\Delta_j+\Vext(t,x_j)\right)+\frac{1}{N-1}\sum\limits_{i<j}\vb(x_i-x_j)\,.
\end{equation}
Here, $\Vext$ denotes some possibly time-dependent external potential, and the interaction potential $\vb$ is defined as
\begin{equation}\label{vb}
\vb(x):=N^{\D\beta}v(N^\beta x)\,,\qquad \beta\in[0,\tfrac{1}{\D})\,,
\end{equation}
for some  bounded, spherically symmetric and compactly supported function $v:\R^\D\to\R$. 
In the following, we will make use of the abbreviation
$$\vb_{ij}:=\vb(x_i-x_j)\,.$$
Note that the prefactor $(N-1)^{-1}$ in front of $\vb$ is chosen such that the interaction energy and the kinetic energy per particle are of the same order.
The mean inter-particle distance is of order $ N^{-\frac{1}{\D}}$ and therefore much smaller than the range of the interaction, which scales as $N^{-\beta}$. 
Hence, on average, every particle interacts with many other particles, and the interactions are weak since $(N-1)^{-1}N^{\D\beta}\to 0$ as $N\to\infty$.
This implies that we consider a mean-field regime. 
In particular, the case $\beta=0$ is known as the Hartree scaling regime.

We study the time evolution of the $N$-body system for large $N$ when the bosons initially exhibit Bose--Einstein condensation.
We impose suitable conditions on the external potential $\Vext(t)$ such that $\Hb(t)$ is self-adjoint on $\mathcal{D}(\Hb(t))=H^2(\R^{\D N})$ for each $t\in\R$. Consequently, $\Hb(t)$ generates a unique family of unitary time evolution operators $\{U(t,s)\}_{t,s\in\R}$ via the Schrödinger equation
\begin{equation}\label{U}
\i\tfrac{\d}{\d t}U(t,s)=\Hb(t)U(t,s)\,,\qquad U(s,s)=\mathbbm{1}\,.
\end{equation}
The $N$-body wave function at time $t\in\R$ is determined by 
\begin{equation}\label{SE}
\psi(t)=U(t,0)\psi(0)
\end{equation}
for some initial datum $\psi(0)=\psi_0\in L^2_\mathrm{sym}(\R^{\D N})$.
Due to the interactions, the characterisation of the time evolution $U(t,s)$ is a difficult problem. Even if the system was initially in a factorised state, where all particles are independent of each other, the interactions instantaneously correlate the particles such that an explicit formula for $U(t,s)$ is quite inaccessible.\\

To describe $U(t,s)$ approximatively, one observes that the dynamics of the many-body system can be decomposed into the dynamics of the condensate wave function $\varphi(t)\in L^2(\R^\D)$ and 
the dynamics of the excitations from the (time-evolved) condensate. 
The evolution of $\pt$ approximates the $N$-body dynamics $\psi(t)$ in the sense of reduced densities (see Section \ref{subsec:leading:order}). 
Moreover, if the dynamics of the excitations are suitably approximated and added to the description, one obtains an approximation of $\psi(t)$ with respect to the $L^2(\R^{\D N})$ norm, in the sense that
\begin{equation*}
\norm{\psi(t)-\psi_\mathrm{approx}(t)}^2_{L^2(\R^{\D N})}\leq C(t) N^{-\delta}
\end{equation*}
for some power $\delta\in(0,1]$ depending on the choice of $\beta$ (see Section \ref{subsec:next:order}).\\

In this paper, we introduce a novel method for deriving  a more precise characterisation of the dynamics.
This is achieved by constructing a sequence of $N$-body wave functions, which are defined via an iteration of Duhamel's formula with the time evolution $\Ubog(t,s)$ generated by an auxiliary Hamiltonian $\Hpt(t)$ (see \eqref{Hpt} for a precise definition).
Under the assumption that sufficiently high moments of the number of excitations in the initial state are subleading, we prove higher order corrections to the norm approximation for the scaling regime $\beta\in[0,\frac{1}{4\D})$.
This is to be understood in the following sense:\  we construct a sequence of $N$-body wave functions $\{\psia(t)\}_{a\in\mathbb{N}}\subset L^2(\mathbb{R}^{\D N})$
such that, for sufficiently large $N$,
\begin{equation}\label{eqn:aim}
\norm{\psi(t)-\psia(t)}^2_{L^2(\R^{\D N})}\leq C(t) N^{-a\delta(\beta,\gamma)}
\end{equation}
for some time-dependent constant $C(t)$.
The positive exponent $\delta(\beta,\gamma)$ is  determined in Theorem~\ref{thm:corrections}.
It depends on $\beta$ and on a parameter $\gamma$ which is related to the initial number of excitations (see assumption A3).

Let us remark that the approximating functions $\psia(t)$ are $N$-body wave functions.
Since they are explicitly given in terms of the dynamics $\Ubog(t,s)$ related to the (first order) norm approximation, the functions $\psia(t)$ are much more accessible than the true dynamics $\psi(t)$.
In particular, all higher order corrections can be obtained from the norm approximation $\Ubog(t,0)\psi_0$ by computing an $N$-independent number of integrals.

Moreover, if the initial excitation vector is quasi-free, one can extend the method introduced in this paper. 
We conjecture that it is possible to  approximate all $n$-point correlation functions of the full dynamics $\psi(t)$ to arbitrary precision by expressions depending only on approximations of the 2-point correlation functions, whose computation reduces to solving two coupled linear one-body equations (\cite[Equations (17a-b)]{grillakis2013} and \cite[Equation (34)]{nam2015}). 
This would mean a huge simplification of the full $N$-body problem \eqref{U} corresponding to the Hamiltonian \eqref{Hb} for certain initial states, in particular with regard to a numerical analysis, and we plan to show this in a separate paper.
\\

Finally, we note that higher order approximations of the reduced density matrices 
were obtained by Paul and Pulvirenti in~\cite{paul2019} for $\beta=0$ and factorised initial data, 
based on the method of kinetic errors from the paper by Paul, Pulvirenti and Simonella~\cite{paul2018}.
For $j\in\{1\mydots N\}$, the authors of \cite{paul2019} construct a sequence $\{F^{N,n}_j(t)\}_{n\in\N}$ of trace class operators on $L^2(\R^{jd})$,
which approximate the $j$-particle reduced density matrix $\gamma^{(j)}(t)$ of the system with increasing accuracy up to arbitrary precision. 
The approximating operators $F^{N,n}_j(t)$ 
can be determined by a number of operations scaling with $n$. They
depend on the initial data as well as the knowledge of the solution of the Hartree equation and its linearisation around this solution.

Due to different methods used, it is not straightforward to compare the results of \cite{paul2019} with the results of this paper. However, we list some features of 
our paper that differ from the operator-based method of kinetic errors \cite{paul2019,paul2018}.  
In contrast to the approach in \cite{paul2019}, we derive approximations directly for the time-evolved $N$-body wave function, and our construction is implemented as a robust algorithm 
that requires an $a$-dependent, $N$-independent number of explicit calculations to compute the $a$'th order approximation.
Moreover, the results obtained in this paper  include positive values of $\beta$ and cover more generic initial data than  \cite{paul2019}, where the initial state is assumed factorised,  i.e., with zero initial excitations.

\bigskip
\noindent\textbf{Notation.} 
In the following, any expression $C$ that is independent of both $N$ and $t$ will be referred to as a constant.
Note that constants may depend on all fixed parameters of the model such as $\pz$, $\psi_0$, $v$ and $\Vext(0)$.
Further, we denote $A\ls B$ and $A\gs B$ to indicate that there exists a constant $C>0$ such that $A\leq C B$, resp. $A\geq CB$, and
abbreviate
$$\lr{\,\cdot\,,\,\cdot\,}_{L^2(\R^{\D N})}=:\lr{\,\cdot\,,\,\cdot\,}, \qquad \norm{\,\cdot\,}_{L^2(\R^{\D N})}=:\norm{\,\cdot\,}, \qquad \norm{\,\cdot\,}_{\mathcal{L}(L^2(\R^{\D N}))}=:\onorm{\,\cdot\,}.$$
Finally, we use the notation
$$\lfloor r\rfloor:=\max\left\{z\in\mathbb{Z}:\,z\leq r\right\}\,,\qquad
\lceil r\rceil:=\min\left\{z\in\mathbb{Z}:\,z > r\right\} $$ 
for $r\in\R$.


\section{Known results}
\subsection{Leading order approximation: reduced densities}\label{subsec:leading:order}
A first approximation to the $N$-body dynamics is provided by the time evolution of the condensate wave function. 
Its dynamics yield a macroscopic description of the Bose gas, which, in the limit $N\to\infty$, coincides with the true dynamics in the sense of reduced density matrices.
In order to formulate this mathematically, one assumes that the system is initially in a Bose--Einstein condensate with condensate wave function $\varphi_0$, i.e., 
$$\lim\limits_{N\to\infty}\Tr\left|\gamma^{(1)}(0)-|\varphi_0\rangle\langle\varphi_0|\right|=0\,,$$
where 
$$\gamma^{(1)}(t):=\Tr_{2\mydots N}|\psi(t)\rangle\langle\psi(t)|$$
is the one-particle reduced density matrix of $\psi(t)$ at time $t$. Then it has been shown,  see e.g. \cite{adbagote, adami2007, chhapase, chen2013_2, erdos2006, erdos2007, kirkpatrick2011, sohinger15}, that
\begin{equation} \label{intro-gamma1-convergence}
\lim\limits_{N\to\infty}\Tr\left|\gamma^{(1)}(t)-|\pt\rangle\langle\pt|\right|=0
\end{equation}
for any $t\in\R$, where $\pt$ is the solution of the Hartree equation
\begin{equation}\label{NLS}
\i\tfrac{\d}{\d t}\varphi(t)
=\left(-\Delta+\Vext(t)+\vbarpt-\mpt\right)\varphi(t)
=:\hpt(t)\varphi(t)
\end{equation}
with initial datum $\varphi(0)=\varphi_0$ and with
\begin{equation}\label{vbarpt}
\vbarpt(x):=\left(\vb\ast|\varphi(t)|^2\right)(x):=\int_{\R^\D}\vb(x-y)|\varphi(t,y)|^2\d y \,.
\end{equation}
Note that for $\beta=0$, the equation~\eqref{NLS} is the $N$-independent Hartree (NLH) equation.
For $\beta>0$, the evolution is $N$-dependent and converges to the non-linear Schr\"odinger (NLS) dynamics with $N$-independent coupling parameter $\int v$ in the limit $N\to\infty$.
The parameter $\mpt$ is a real-valued phase factor, which we choose as
\begin{equation}\label{mpt}
\mpt:=\tfrac{1}{2}\int_{\R^\D}\d x|\varphi(t,x)|^2\,\vbarpt(x)\; = \;
\tfrac{1}{2}\int_{\R^\D}\d x\int_{\R^\D}\d y\,|\varphi(t,x)|^2|\varphi(t,y)|^2\vb(x-y)
\end{equation}
for later convenience. For the convergence with respect to reduced densities, this phase is irrelevant since it cancels in the projection $|\pt\rangle\langle\pt|$.

One way to prove the convergence \eqref{intro-gamma1-convergence}, and consequently to derive the NLH/ NLS equation from a system of $N$ bosons, is via the so-called BBGKY\footnote{(Bogoliubov-Born-Green-Kirkwood-Yvon)} hierarchy, which was prominently used in the works of
Lanford for the study of classical mechanical systems in the infinite particle limit \cite{la-1968,la-1968/69}.
The first derivation of the NLH equation via the BBGKY hierarchy was given
by Spohn in \cite{spohn1980}, and this was further pursued, e.g., in  \cite{adbagote,adami2007,froehlich2007, frknsc}. 
About a decade ago, Erd\H{o}s, Schlein and Yau 
fully developed the  BBGKY hierarchy approach to the derivation of 
the NLH/NLS equation in their seminal works including \cite{erdos2006,erdos2007}. 
Subsequently, a crucial step of this method was revisited by 
Klainerman and Machedon in \cite{klainerman2008}, based on reformulating combinatorial argument in \cite{erdos2006,erdos2007}
and a viewpoint inspired by methods of non-linear PDEs. This, in turn, motivated many recent works on the derivation of dispersive PDEs, 
including \cite{chhapase, chpa, chen2013_2,chen2016, chen2017_2, kirkpatrick2011, sohinger15}.
In~\cite{rodnianski2009}, Rodnianski and Schlein introduced yet another method for proving~\eqref{intro-gamma1-convergence}, which uses coherent states on Fock space and was inspired by 
techniques of quantum field theory and the pioneering work of Hepp~\cite{hepp}.

In the context of the current paper, the most relevant works on the derivation of the NLH/NLS equation are due to Pickl~\cite{pickl2011,pickl2015}, who introduced an efficient method for deriving effective equations from the many-body dynamics, transforming the physical idea behind the mean-field description of an $N$-body system into a mathematical algorithm. 
Instead of describing the condensate as the vacuum of a Fock space of excitations, this approach remains in the $N$-body setting and uses projection operators to factor out the condensate.
This strategy was successfully applied to prove effective dynamics for $N$-boson systems in various situations, e.g.,~\cite{anapolitanos2016, GP,deckert2014,jeblick2016, jeblick2017, knowles2010, michelangeli2017_2, michelangeli2017}.

\subsection{Next-to-leading order: norm approximation}\label{subsec:next:order}
Whereas closeness in the sense of reduced densities implies that the majority of the particles (up to a relative number that vanishes as $N\to\infty$) is in the state $\pt$, 
the norm approximation requires the control of all $N$ particles.
In particular, this implies that the excitations from the condensate can no longer be omitted from the description.
In this sense, the norm approximation of $\psi(t)$ can be understood as next-to-leading order correction to the mean-field description.

A norm approximation for initial coherent states on Fock space was first obtained in \cite{grillakis2010,grillakis2011} by Grillakis, Machedon and Margetis. 
For initial states $\psi_0\in L^2(\R^{dN})$ with fixed particle number, Lewin, Nam and Schlein proved in~\cite{lewin2015} a norm approximation for $\beta=0$ and $\Vext=0$ under quite general assumptions on the interaction potential $v$.
Nam and Napi\'orkowski extended this result in~\cite{nam2015} to the range $\beta\in[0,\frac13)$, in \cite{nam2017} to the range $\beta\in[0,\frac12)$ for the three-dimensional defocusing case, 
and in~\cite{nam2017_2} to the focusing case in dimensions one and two for $\beta>0$ and $\beta\in(0,1)$, respectively.
As proposed in \cite{lewin2015_2}, the authors decomposed the $N$-body wave function $\psi(t)$ into condensate and excitations as
\begin{equation}\label{eqn:orth:excitations}
\psi(t)=\sum\limits_{k=0}^N\pt^{\otimes (N-k)}\otimes_s\xiptk
\end{equation}
for some $\xipt=\big(\xiptk\big)_{k=0}^N\in\FNpt$, 
where 
\begin{equation}\label{FNpt}
\FNp:= \bigoplus_{k=0}^N\bigotimes_\mathrm{sym}^k \{\varphi\}^\perp
\end{equation}
is the truncated bosonic Fock space over the orthogonal complement in $L^2(\R^\D)$ of the span of $\varphi\in L^2(\R^\D)$. 
A definition of $\xiptk$ will be given in~\eqref{eqn:phik}.
Further, $\otimes_s$ denotes the symmetric tensor product, which is for $\psi_a\in L^2(\R^{\D a})$, $\psi_b\in L^2(\R^{\D b})$ defined as 
\begin{equation*}
(\psi_a\otimes_s\psi_b)(x_1\mydots x_{a+b}):= \frac{1}{\sqrt{a!\,b!\,(a+b)!}}\sum\limits_{\sigma\in \mathfrak{S}_{a+b}}\psi_a(x_{\sigma(1)}\mydots x_{\sigma(a)})\,\psi_b(x_{\sigma(a+1)}\mydots x_{\sigma(a+b)})\,,
\end{equation*}
where $\mathfrak{S}_{a+b}$ denotes the set of all permutations of $a+b$ elements.
The addend $k=0$ in~\eqref{eqn:orth:excitations} describes the condensate, while the terms $k\in\{1\mydots N\}$ correspond to the excitations.
In the following, we will  refer to $\xipk(t)$ as \emph{$k$-particle excitation}.

In \cite{lewin2015,nam2015,nam2017_2,nam2017}, the authors consider initial data of the form
\begin{equation}\label{eqn:initial:cond:norm:approx}
\psi_0=\sum\limits_{k=0}^N\pz^{\otimes(N-k)}\otimes_s\chik(0)
\end{equation}
for some appropriate initial excitation vector $\chi(0):=(\chik(0))^\infty_{k=0}\in\mathcal{F}(\{\pz\}^\perp)$.
It is then shown that there exist constants $C,C'>0$ such that
\begin{equation}\label{eqn:norm:approx}
\left\|\psi(t)-\sum\limits_{k=0}^N\pt^{\otimes(N-k)}\otimes_s\chik(t)\right\|^2_{L^2(\R^{\D N})} \leq  C \e^{C't}N^{-\delta}\,,
\end{equation}
where $\delta=1-3\beta$ for the three-dimensional defocusing case with $\beta\in[0,\frac13)$
and $\delta=\frac12$ and $\delta<\frac13(1-\beta)$ for the one- and two-dimensional focusing case, respectively.
The excitations  $\chi(t)=(\chik(t))_{k= 0}^\infty \in \mathcal{F}(\{\pt\}^\perp)$ at time $t>0$
are determined by the Bogoliubov evolution,
\begin{equation}\label{eqn:H:Bog}
\i\tfrac{\d}{\d t}\chi(t)=\mathbb{H}_\mathrm{Bog}(t)\chi(t)\,.
\end{equation}
Here, $\mathbb{H}_\mathrm{Bog}(t)$ denotes the Bogoliubov Hamiltonian%
\footnote{
	Written in second quantized form, $\Hbog(t)$ is defined as
	\begin{equation*}\begin{split}
	\Hbog(t):=\int_{\R^d}a^*_x\left(\hpt(t,x)+K_1(t,x)\right) a_x\d x
	+\tfrac12\int_{\R^d}\d x\int_{\R^d}\d y\left(K_2(t,x,y)a^*_x a^*_y+\overline{K_2(t,x,y)}a_xa_y\right)\,,
	\end{split}\end{equation*}
	where $a^*_x$ and $a_x$ denote the operator-valued distributions corresponding to the usual creation and annihilation operators on $\mathcal{F}(L^2(\R^d))$.
	Besides, $K_1(t):=Q(t)\tilde{K}_1(t)Q(t)$ with $Q(t):=1-|\pt\rangle\langle\pt|$,
	where $\tilde{K}_1$ is the Hilbert-Schmidt operator on $L^2(\R^d)$ with kernel 
	$\tilde{K}_1(t,x,y):=\varphi(t,x)\vb(x-y)\overline{\varphi(t,y)}$.
	Further, $K_2(t):=\left(Q(t)\otimes Q(t)\right)\tilde{K}_2(t,\cdot,\cdot)$, 
	where 
	$\tilde{K}_2(t,x,y):=\varphi(t,x)\vb(x-y)\varphi(t,y)$ (e.g.~\cite[Equation~(31)]{nam2015}).
},
an effective  Hamiltonian on Fock space which is quadratic in the number of creation and annihilation operators.\\

For three dimensions and  scaling parameter $\beta=0$, a similar result was obtained by Mitrouskas, Petrat and Pickl in~\cite{mitrouskas2016, mitrouskas_PhD} via a first quantised approach, where the splitting of $\psi(t)$ into condensate and excitations is realised by means of projections as introduced  in~\cite{pickl2011}.
Since we will work in the first quantised setting, let us  recall this approach and introduce some notation.
\begin{definition}\label{def:p:q}
Let $\varphi\in L^2(\R^\D)$. 
Define the orthogonal projections on $L^2(\R^\D)$
$$\pp:=|\varphi\rangle\langle\varphi|, \qquad \qp:=\mathbbm{1}-p^\varphi$$
and the corresponding projection operators on $L^2(\R^{\D N})$
$$\pp_j:=\underbrace{\mathbbm{1}\otimes\cdots\otimes\mathbbm{1}}_{j-1}\otimes\, \pp\otimes \underbrace{\mathbbm{1}\otimes\cdots\otimes\mathbbm{1}}_{N-j} \quad\text{and}\quad \qp_j:=\mathbbm{1}-\pp_j\,.$$
For $0\leq k\leq N$, define the many-body projections
\begin{eqnarray*}
\Pp_k
&:=&\sum\limits_{\substack{J\subseteq\{1,\dots,N\}\\|J|=k}}\prod\limits_{j\in J}q^\varphi_j\prod\limits_{l\notin J}p^\varphi_l
\;=\;\frac{1}{(N-k)!k!}\sum\limits_{\sigma\in\mathfrak{S}_N}\qp_{\sigma(1)}\mycdots \qp_{\sigma(k)}\pp_{\sigma(k+1)}\mycdots \pp_{\sigma(N)}
\end{eqnarray*}
and $\Pp_k=0$ for $k<0$ and $k>N$.
Further, for any function $f: \N_0\rightarrow\R_0^+$ and any $j\in\mathbb{Z}$, define the operators $\hat{f^\varphi}, \hat{f^\varphi_j}\in\mathcal{L}\left(L^2(\R^{\D N})\right)$ by 
$$\hat{f^\varphi}:= \sum\limits_{k=0}^N f(k)\Pp_k, \qquad 
\hat{f^\varphi_j}:=\sum\limits_{n=-j}^{N-j} f(n+j)\Pp_n\,.$$
We will in particular need the operators $\hat{n^\varphi}$ and $\hat{m^\varphi}$ corresponding to the weights
$$n(k):=\sqrt{\tfrac{k}{N}}\,, \qquad m(k):=\sqrt{\tfrac{k+1}{N}}\,.$$
\end{definition}
Hence, for any $\psi\in L^2(\R^{\D N})$, the part of $\psi$ in the condensate $\varphi^{\otimes N}$ is given by $\Pp_0\psi$, and the part of $\psi$ corresponding to $k$ particles being excited from the condensate is precisely $\Pp_k\psi$ for $k\geq1$. By construction,  $\Pp_k\Pp_{k'}=\delta_{k,k'}\Pp_k$.
Besides, the identity $\sum_{k=0}^N \Pp_k=\mathbbm{1}$ implies
\begin{equation}\label{eqn:decomposition}
\psi\;=\;\sum_{k=0}^N \Pp_k\psi \;=:\;\sum_{k=0}^N \varphi^{\otimes (N-k)}\otimes_s\xipk
\end{equation}
for some $\xipk\in L^2(\R^{\D k})$.
To determine the explicit form of $\xipk$, observe that by Definition~\ref{def:p:q},
\begin{eqnarray*}
\Pp_k\psi(x_1\mydots x_N)
&=&\frac{1}{(N-k)!k!}\sum\limits_{\sigma\in\mathfrak{S}_N}
	\varphi(x_{\sigma(k+1)})\mycdots\varphi(x_{\sigma(N)})\,\qp_{\sigma(1)}\mycdots \qp_{\sigma(k)} \times\\
	&&\times\, \int\limits_{\R^d}\d y_1\mycdots \int\limits_{\R^d}\d y_{N-k}\,\overline{\varphi(y_1)}\mycdots\overline{\varphi(y_{N-k})}\,\psi(x_{\sigma(1)}\mydots x_{\sigma(k)},y_1\mydots y_{N-k})\\
&=:&\left(\varphi^{\otimes (N-k)}\otimes_s\xipk\right)(x_1\mydots x_N)\,,
\end{eqnarray*}
where, by definition of the symmetric tensor product,
\begin{eqnarray}
&&\hspace{-1cm}\xipk(x_1\mydots x_k):=\nonumber\\
&&=\sqrt{\tbinom{N}{k}}\, \qp_{1}\mycdots \qp_{k}
\int\limits_{\R^d}\d \tilde{x}_{k+1}\mycdots \int\d \tilde{x}_N\,\overline{\varphi(\tilde{x}_{k+1})}\mycdots\overline{\varphi(\tilde{x}_N)}\,\psi(x_1\mydots x_k,\tilde{x}_{k+1}\mydots\tilde{x}_N)\,.\label{eqn:phik}
\end{eqnarray}
Obviously, $\xipk$ is symmetric under permutations of all of its coordinates,
and $\xipk$ is orthogonal to $\varphi$ in every coordinate, i.e.,
\begin{equation}\label{eqn:phik:orth}
\int_{\R^\D}\overline{\varphi(x_j)}\,\xipk(x_1\mydots x_j\mydots x_N) \d x_j=0\,,\qquad \pp_j\xipk = 0 \,,\qquad \qp_j\xipk=\xipk
\end{equation}
for every $j\in\{1\mydots k\}$. Hence, $\xipk\in\bigotimes_\mathrm{sym}^k\{\varphi\}^\perp$.
The excitations $\xipk$, $k\in\{0\mydots N\}$, define a vector $\xip:=\left(\xipz,\xipo\mydots \xipN\right) $ in the truncated Fock space $\FNp$ defined in~\eqref{FNpt}.
The relation between the $N$-body state $\psi$ and the corresponding excitation vector $\xip$ is given by the unitary map
\begin{eqnarray}\label{eqn:map:U}
\UNp:L^2(\R^{\D N}) \to  \FNp\;, \quad
\psi  \mapsto  \UNp\psi:=\xip\,,
\end{eqnarray}
where $\xip$ is defined by~\eqref{eqn:phik}.
The vacuum $(1,0\mydots 0)$ of $\FNp$ corresponds to the condensate $\varphi^{\otimes N}$, 
and the probability of $k$ particles being outside the condensate equals
\begin{equation}\label{eqn:phik:norm}
\norm{\xipk}^2_{L^2(\R^{\D k})}=\tbinom{N}{k}\norm{\qp_1\cdots \qp_k \pp_{k+1}\cdots \pp_N\psi}^2=\norm{\Pp_k\psi}^2
\end{equation}
by~\eqref{eqn:phik}.
The number operator $\Np$ on $\FNp$, counting the number of excitations, is defined by its action
$$\left(\Np\,\xip\right)^{(k)}:=k\,\xipk\,.$$
The expected number of excitations from the condensate $\varphi^{\otimes N}$ in the state $\psi$ is thus given by
\begin{equation}\begin{split}\label{eqn:N}
\lr{\xip,\Np\,\xip}_\FNp&=\sum\limits_{k=0}^N k\norm{\xipk}^2_{L^2(\R^{\D k})}=\sum\limits_{k=0}^N k\norm{\Pp_k\psi}^2=
N\lr{\psi,\sum\limits_{k=0}^N\tfrac{k}{N}\Pp_k\psi}\\
&=N\norm{\hat{n^\varphi}\psi}^2
\end{split}\end{equation}
with $\hat{n^\varphi}$ from Definition~\ref{def:p:q}.\\

In \cite{mitrouskas2016}, the authors introduce an auxiliary $N$-particle Hamiltonian $\Hpt(t)$ by subtracting from $\Hb(t)$ in each coordinate the mean-field Hamiltonian $\hpt(t)$ from~\eqref{NLS}, inserting identities 
$$(\ppt_i+\qpt_i)(\ppt_j+\qpt_j)$$ 
on both sides of the difference, and discarding all terms which are cubic, $\Cpt$, or quartic, $\Qpt$, in the number of projections $\qpt$:
\begin{lem}\label{lem:Hpt}
$$\Hb(t)=\Hpt(t)+\Cpt+\Qpt\,,$$
where
\begin{equation}\begin{split}\label{Hpt}
\Hpt(t)\;:=\;&\sum\limits_{j=1}^N\hpt_j(t)\\
&+\frac{1}{N-1}\sum\limits_{i<j}\left(\ppt_i\qpt_j\vb_{ij}\qpt_i\ppt_j+\ppt_i\ppt_j\vb_{ij}\qpt_i\qpt_j+\mathrm{h.c.}\right)\,,
\end{split}\end{equation}
and with
\begin{eqnarray*}
\Cpt&:=&\frac{1}{N-1}\sum\limits_{i<j}\bigg(\qpt_i\qpt_j\Big(\vb_{ij}-\vbarpt(x_i)-\vbarpt(x_j)\Big)(\qpt_i\ppt_j+\ppt_i\qpt_j)+\mathrm{h.c.}\bigg)\,, \\
\Qpt&:=&\frac{1}{N-1}\sum\limits_{i<j}\qpt_i\qpt_j\left(\vb_{ij}-\vbarpt(x_i)-\vbarpt(x_j)+2\mpt\right)\qpt_i\qpt_j\,.
\end{eqnarray*}
\end{lem}
\begin{proof}
\begin{eqnarray*}
\Hb(t)&=&\sum\limits_{j=1}^N\hpt_j(t)+\frac{1}{N-1}\sum\limits_{i<j}\vb_{ij}-\sum\limits_{j=1}^N\vbarpt(x_j)+N\mpt\\
&=&\sum\limits_{j=1}^N\hpt_j(t)+\frac{1}{N-1}\sum\limits_{i<j}\left(\vb_{ij}-\vbarpt(x_i)-\vbarpt(x_j)+2\mpt\right)\,.
\end{eqnarray*}
Now one inserts identities $\mathbbm{1}=(\ppt_i+\qpt_i)(\ppt_j+\qpt_j)$ before and after the expression in the brackets and uses the relations
$$\ppt_i\vb_{ij}\ppt_i=\vbarpt(x_j)\ppt_i, \qquad 
\ppt_i\vbarpt(x_i)\ppt_i=2\mpt \ppt_i\,,$$
which concludes the proof.
\end{proof}
The auxiliary Hamiltonian $\Hpt(t)$ has a quadratic structure comparable to that of the Bogoliubov-Hamiltonian $\Hbog(t)$: all terms in $\Hpt(t)-\sum_j\hpt_j(t)$, which form an effective two-body potential, contain exactly two projectors $\qpt$ onto the complement of the condensate wave function, while $\Hbog(t)$ is quadratic in the creation and annihilation operators of the excitations.
However, $\Hpt(t)$ is particle number conserving and acts on the $N$-body Hilbert space $L^2(\R^{\D N})$, 
i.e., it determines the evolution of the $N$-body wave function consisting of excited particles and particles in the condensate $\pt$, with $\pt$ the solution of \eqref{NLS}.
In contrast, $\Hbog(t)$ operates on the excitation Fock space $\Fpt$, does not conserve the particle number, and exclusively concerns the dynamics of the excitations with respect to the condensate wave function evolving according to~\eqref{NLS}.

The time evolution generated by $\Hpt(t)$ is denoted by $\Ubog(t,s)$.
For an initial datum $\psi_0\in L^2_\mathrm{sym}(\R^{\D N})$, the corresponding $N$-body wave function at time $t\in\R$ is
\begin{equation}\label{psibog}
\psibog(t)=\Ubog(t,0)\psi_0\,.
\end{equation}
Existence and uniqueness of of the time evolution $\Ubog(t,s)$ are recalled in Lemma \ref{lem:Bog}.

Under appropriate assumptions on the initial datum $\psi_0$, the time evolution $\Ubog(t,s)$ approximates the actual time evolution $U(t,s)$. More precisely, there exist constants $C,C'>0$ such that
\begin{equation}\label{eqn:David:norm}
\big\|\big(U(t,0)-\Ubog(t,0)\big)\psi_0\big\|^2_{L^2(\R^{\D N})}\leq C\e^{C't^2}N^{-1}
\end{equation}
\cite[Theorem 2.6]{mitrouskas2016}.
Further, in the limit $N\to\infty$, the excitations in $\Ubog(t,0)\psi_0$ coincide with the solutions of the Bogoliubov evolution equation: 
let $\xipZ=\big(\xipzk\big)_{k=0}^N$ denote the excitations from $\pz^{\otimes N}$ in the initial state $\psi_0$ under the decomposition~\eqref{eqn:orth:excitations}, let 
$\xitpt=\big(\xitptk\big)^N_{k=0}$ denote the excitations from $\pt^{\otimes N}$ in $\Ubog(t,0)\psi_0$, and let $\chi(t)=\left(\chik(t)\right)_{k\geq 0}$ denote the solutions of~\eqref{eqn:H:Bog} with initial datum $\chi^{(k)}(0)=\xipZ^{(k)}$ for $0\leq k\leq N$ and $\chi^{(k)}(0)=0$ for $k>N$.
Then
\begin{equation}\label{eqn:David:equiv}
\sum\limits_{k=0}^N\Big\|\xitptk-\chik(t)\Big\|^2_{L^2(\R^{\D k})}\leq C\e^{C't^2}N^{-1}
\end{equation}
\cite[Lemma 2.8]{mitrouskas2016}.
Hence, the combination of~\eqref{eqn:David:norm} and~\eqref{eqn:David:equiv} yields~\eqref{eqn:norm:approx}, with a different time-dependent constant but the same $N$-dependence.\\

Finally, let us remark that for larger values of the scaling parameter, beyond the mean field regime, the evolutions of $\pt$ and $\xipt$ do not (approximately) decouple any more as a consequence of the short-scale structure related to the two-body scattering process.
For $\beta\in(0,1)$, an accordingly adjusted variant of~\eqref{eqn:norm:approx} for appropriately modified initial data was proved by Brennecke, Nam, Napi\'orkowski and Schlein in~\cite{brennecke2017_2} in the three-dimensional defocusing case.
Similar estimates for the many-body evolution of appropriate classes of Fock space initial data have been obtained in~\cite{boccato2015, chong2016, ginibre1979, ginibre1979_2, grillakis2013, grillakis2017, grillakis2010, grillakis2011, kuz2017, rodnianski2009} for various ranges of the scaling parameter.
A related result for Bose gases with large volume and large density was proved in~\cite{petrat2017}.


\section{Main results}\label{sec:main}
\subsection{Assumptions}\label{subsec:main:framework}
Let us  state our assumptions on the model~\eqref{Hb} and on the initial data.
\begin{itemize}
\item[A1] \emph{Interaction potential.} 
Let $v:\R^\D\to\R$ be spherically symmetric and bounded uniformly in $N$, i.e., $\norm{v}_{L^\infty(\R^\D)}\ls 1$. Further, assume that 
$\supp v\subseteq\{x\in\R^\D:|x|\ls 1\}$.
\item[A2] \emph{External potential.}
Let $\Vext:\R\times\R^\D\to\R$ such that $\Vext(\cdot,x)\in\mathcal{C}(\R)$ for each $x\in\R^d$ and $\Vext(t,\cdot)\in L^\infty(\R^\D)$ for each $t\in\R$.
\item[A3] \emph{Initial data.} 
Let  $\psi_0\in H^2(\R^{\D N})\cap L^2_\mathrm{sym}(\R^{\D N})$ and $\varphi_0\in H^k(\R^\D)$, $k=\lceil\tfrac{\D}{2}\rceil$, both be normalised. Let $\gamma\in(0,1]$ and $A\in\N$.
Assume that for any $a\in\{0\mydots A\}$, there exists a set of non-negative, $a$-dependent constants $\left\{\fcaz\right\}_{0\leq a\leq A}$ with $\fczz=1$ such that, for sufficiently large $N$,
$$\left\|\left(\hat{m^{\varphi_0}}\right)^a \psi_0\right\|^2\leq \fcaz\,N^{-\gamma a}\,.$$
\end{itemize}
Our analysis is valid as long as the solution $\pt$ of the non-linear equation~\eqref{NLS} exists in $H^k(\R^\D)$-sense for $k=\lceil\tfrac{\D}{2}\rceil$. 
The maximal time of $H^k(\R^d)$-existence, $\Tex$, is defined as
\begin{equation}\label{existence:time}
\Tex:=\sup\left\{ t\in \R_0^+:\norm{\pt}_{H^k(\R^\D)}<\infty \text{ for } k=\lceil\tfrac{\D}{2}\rceil\right\}
\end{equation}
and depends on the dimension $\D$, the sign of $\vbarpt$, and the regularity of the external trap $\Vext (t)$.
Under assumptions A1 and A2 and for times $s,t\in[0,\Tex)$, the time evolution $\Ubog(t,s)$ is well-defined.

\begin{lem}\label{lem:Bog}
Let $s,t\in\big[0,\Tex\big)$. Then $\Hpt(t)$ is self-adjoint on $\mathcal{D}(\Hpt(t))=H^2(\R^{\D N})$ and  generates a unique family of unitary time evolution operators $\Ubog(t,s)$. $\Ubog(t,s)$ is strongly continuous jointly in $s,t$ and leaves $H^2(\R^{\D N})$ invariant.
\end{lem}
\begin{proof}
As a consequence of the Sobolev embedding theorem (e.g.\ \cite[Theorem 4.12, Part IA]{adams}), $\norm{\pt}_{L^\infty(\R^\D)}\ls \norm{\pt}_{H^k(\R^\D)}$ for $k=\lceil\tfrac{\D}{2}\rceil$.
Hence, by definition~\eqref{existence:time} of $\Tex$, $\mpt$ and $(N-1)\vbarpt$ are bounded uniformly in $N$ for $t\in\big[0,\Tex\big)$.
Further, $t\mapsto\Hpt(t)\psi$ is Lipschitz for all $\psi\in H^2(\R^{\D N})$ because of~\eqref{NLS}, since $t\mapsto\Vext(t)\in\mathcal{L}(L^2(\R^\D))$ is continuous and as $\tfrac{\d}{\d t}\ppt=\i[\ppt,\hpt(t)]$. Hence, the statement of the lemma follows from~\cite{griesemer2017}.
\end{proof}

Assumptions A1 and A2 are rather standard in the rigorous treatment of interacting many-boson systems. Note that we make no assumption on the sign of the potential or its scattering length and thus cover both repulsive and attractive interactions.
Besides, we admit a large class of time-dependent external traps $\Vext$, with basically the only restriction that $\Vext(t)$ must not obstruct the self-adjointness of $\Hb(t)$ on $H^2(\R^{\D N})$. 

The third assumption provides a bound on the expected number of excitations from the condensate $\pz^{\otimes N}$ in the initial state $\psi_0$. 
Note that while $\gamma=0$ is the trivial bound, the condition becomes more restrictive as $\gamma$ increases.
We have chosen this particular formulation of A3 for later convenience\footnote{
	Note that the operators $\hat{n^\varphi}$ and $\hat{m^\varphi}$ are equivalent in the sense that they are related via~\eqref{eqn:n:m:op}, namely $(\hat{n^\varphi})^{2a}\leq (\hat{m^\varphi})^{2a} \leq 2^a (\hat{n^\varphi})^{2a}+N^{-a}$, hence all results in terms of $\hat{m^\varphi}$ can be translated to the corresponding statements in terms of $\hat{n^\varphi}$. We chose to work with $\hat{m^\varphi}$ instead of $\hat{n}^\varphi$ because this makes in particular Proposition~\ref{thm:alpha} easier to write. For example, in terms of $\hat{n^\varphi}$, Proposition~\ref{thm:alpha:2} reads\\\vspace{-5pt}
	$$\norm{(\hat{n^\varphi})^j\Ubog(t,s)\psi}^2\ls \cjts\sum\limits_{n=0}^j N^{n(-1+\D\beta)}\left(2^{j-n}\norm{(\hat{n^\varphi})^{j-n}\psi}^2+N^{-j+n}\right)\,,$$
	which contains an additional term $N^{-j+n}$.
	Since the proof of our main result requires an iteration of this proposition, the version with $\hat{m^\varphi}$ is more practicable. 
}. However, its physical meaning is better understood from one of the following two equivalent versions of A3:

\begin{itemize}
\item[A3$\,^\prime$]
Let  $\psi_0\in  H^2(\R^{\D N})\cap L^2_\mathrm{sym}(\R^{\D N})$ 
and $\varphi_0\in H^k(\R^\D)$, $k=\lceil\tfrac{\D}{2}\rceil$, both be normalised. Let $\gamma\in(0,1]$ and $A\in\N$.
Assume that for any $a\in\{0\mydots A\}$, there exists a set of non-negative, $a$-dependent constants $\{\fcaz^\prime\}_{0\leq a\leq A}$ with $\fczz^\prime=1$ such that, for sufficiently large $N$,
$$\left\|\qpz_1\cdots \qpz_a \psi_0\right\|^2\leq \fcaz^\prime\,N^{-\gamma a}.$$
\item[A3$\,^{\prime\prime}$]
Let  $\psi_0\in H^2(\R^{\D N})\cap  L^2_\mathrm{sym}(\R^{\D N})$ 
and $\varphi_0\in H^k(\R^\D)$, $k=\lceil\tfrac{\D}{2}\rceil$, both be normalised. Let $\gamma\in(0,1]$, $A\in\N$ and $\xipZ=\UNpz\psi_0$.
Assume that for any $a\in\{0\mydots A\}$, there exists a set of non-negative, $a$-dependent constants $\{\fcaz^{\prime\prime}\}_{0\leq a\leq A}$ with $\fczz^{\prime\prime}=1$ such that, for sufficiently large $N$,
$$\lr{\xipZ,\Npz^a\,\xipZ}_\FNpz=\sum\limits_{k=0}^N k^a\norm{\xipzk}^2_{L^2(\R^{\D k})}\leq \fcaz^{\prime\prime}\,N^{(1-\gamma)a}\,.$$
\end{itemize}
The equivalence $\mathrm{A3}\Leftrightarrow\mathrm{A3}\,^\prime\Leftrightarrow\mathrm{A3}\,^{\prime\prime}$ follows immediately from Lemma~\ref{lem:equivalence:q:m}, whose proof is postponed to Section~\ref{subsec:preliminaries}.

\begin{lem} \label{lem:equivalence:q:m}
Let $a\in\{1\mydots N\}$ and $\varphi\in L^2(\R^\D)$. Let $\psi\in L^2_\mathrm{sym}(\R^{\D N})$ and $\xip=\UNp\psi$. Then
\lemit{
\item 
$\big\|\qp_1\mycdots\qp_a\psi\big\|^2\;\leq\; \left\|\left(\hat{m^\varphi}\right)^a\psi\right\|^2 \;\leq\;
 4^a a!\,  \sum\limits_{j=1}^aN^{-a+j}\big\|\qp_1\mycdots\qp_j\psi\big\|^2+N^{-a}\,,$
\item $\lr{\xip,\Np^a\,\xip}_\FNp\;\leq\; N^a \left\|\left(\hat{m^\varphi}\right)^a\psi\right\|^2\;\leq\;  1+2^a \lr{\xip,\Np^a\,\xip}_\FNp $ .
}
\end{lem}

Hence, A3 can be understood as follows: 
Let $A\in\N$ and consider sufficiently large $N$ such that $A=\mathcal{O}(1)$ with respect to $N$, i.e.\ $A\ls 1$.
Then we assume that for any $a\leq A$, the part of the wave function with any $a$ particles outside the condensate is at most of order $N^{-\gamma a}$.

Equivalently, A3 states that the first $A\ls1$ moments of the number of excitations must be sub-leading with respect to the particle number; for $\gamma=1$, they must even be bounded uniformly in $N$.
Here, ``sub-leading'' means that the moments of the relative number of excitations, i.e., the expectation values of $(\Npt/N)^A$, vanish as $N\to\infty$.
This, in turn, provides a bound on the high components of the excitation vector: for example, $\sum_{k=0}^N k^A\norm{\xipzk}^2_{L^2(\R^{\D k})}\ls N^{(1-\gamma)A}$ implies $\norm{\xipzN}^2_{L^2(\R^{\D k})}\ls N^{-\gamma A}$. In other words, it must be very unlikely that significantly many particles are outside the condensate, whereas we impose no restriction on excitations involving only few particles (with respect to $N$).

As soon as $a$ becomes comparable to $N$, i.e., $a\gs N$, the constants $\fC_{a}^{(\prime,\prime\prime)}$ are $N$-dependent and the assumption is trivially satisfied. However, note that we demand that $N$ be large enough that $A\ls 1$.
\\

The simplest example of an $N$-body state satisfying A3 is the product state $\psi=\pz^{\otimes N}$.
Whereas the ground state of non-interacting bosons ($v=0$) is of this form, the ground state as well as the lower excited states of interacting systems are not close to an exact product with respect to the $L^2(\R^{\D N})$-norm due to the correlation structure related to the interactions.

Regarding interacting bosons, A3 is fulfilled for quasi-free states with subleading expected number of excitations, since for any quasi-free state $\xi\in\mathcal{F}$ and any  $a\geq 1$ there exists a constant $C_a>0$ such that 
$$\lr{\chi,\mathcal{N}^a\chi}_\mathcal{F}\leq C_a\left(1+\lr{\chi,\mathcal{N}\chi}_\mathcal{F}\right)^a$$
(e.g.\ \cite[Lemma 5]{nam2015}).
In~\cite[Theorem A.1]{lewin2015_2}, it was shown that the ground state of $\Hbog$ is a quasi-free state, which, via the map $\UNp$, defines an $N$-body state $\psi_\mathrm{Bog}$ that converges to the actual ground state $\psi_0$ in norm as $N\to\infty$ \cite[Theorem 2.2]{lewin2015_2}.
Note that  we require a certain minimal size of $\gamma$, which is strictly greater than $\tfrac{2}{3}$.

Further, let us remark that assumption A3  for $A=1$ means complete BEC, which was shown to be a sufficient condition for the validity of the Bogoliubov approximation \cite{lewin2015_2}.
It was shown in \cite[Lemma 1]{seiringer2011} and \cite[Lemma 1]{grech2013} that A3 with $A=1$ and $\gamma=1$ is satisfied for low-energy states of a $d$-dimensional Bose gas for $\beta=0$ in the homogeneous and inhomogeneous setting, respectively. For the Gross--Pitaevskii scaling $\beta=1$ in three dimensions, a comparable bound was proved in \cite{boccato2017, boccato2018_2}.

Besides, A3 with parameter $\gamma<1$ is comparable to part of the assumption made in \cite{nam2017}. In this work, the authors assume that the initial excitation vector $\xi_\pz$ be a quasi-free state in $\mathcal{F}_{\perp\pt}$ such that 
$$\lr{\xi_\pz,\cN_\pz\xi_\pz}\leq \kappa_\varepsilon N^\varepsilon \qquad\text{and}\qquad \lr{\xi_\pz,\d\Gamma(\mathbbm{1}-\Delta)\xi_\pz} \leq \kappa_\varepsilon N^{\beta+\varepsilon}$$
for all $\varepsilon>0$ and with $\kappa_\varepsilon>0$  independent of $N$, and prove a norm approximation for $\beta\in[0,\frac12)$ with parameter $\delta=(1-\varepsilon-2\beta)/2$.
Since $\xi_\pz$ is quasi-free, the first part of this assumption is comparable to   A3, whereas we do not impose any condition on the initial energy of the excitations.

Finally, Mitrouskas showed in~\cite[Chapter 3]{mitrouskas_PhD} that assumption A3 with $\gamma=1$ (and consequently for all $\gamma\in(0,1]$) is fulfilled by the ground state and lower excited states of a homogeneous Bose gas on the $d$-dimensional torus for $\beta=0$.
More precisely, let $\pz$ be the minimizer of the Hartree functional on the torus with ground state energy $E_0$, and let $\psi_n$ denote the $n$'th excited state with energy $E_n$. 
Then the author proves that there exist constants $C,D>0$ such that 
$\norm{\Ppz_a\psi_n}^2\leq C\e^{-Da}$
for all $(E_n-E_0)\leq a\leq N$.
As a corollary of this statement, it is shown that there exists $C_a>0$ such that
$$\lr{\psi_n,\qpz_1\cdots \qpz_a \psi_n}\leq N^{-a}C_a\left(1+\left(E_n-E_0\right)^a\right),
$$
which implies that assumption A3$\,^\prime$ is satisfied.

\subsection{Control of higher moments of the number of excitations}\label{subsec:main:moments}
In our first result, we estimate the growth of the first $A$ moments of the number of excitations when the system
evolves under the dynamics $U(t,s)$ or $\Ubog(t,s)$. 
Estimates of this kind are often needed to derive effective descriptions of the dynamics of interacting bosons, e.g., in~\cite{benarous2013,boccato2015,chen2011,mitrouskas2016,petrat2017,rodnianski2009}.
Our proof extends comparable statements for $\beta=0$ and $\D=3$ obtained in~\cite[Lemma 2.1]{mitrouskas2016} and~\cite[Proposition 3.3]{rodnianski2009}, and for Bose gases with large volume and large density in~\cite[Corollary 4.2]{petrat2017}. 
The estimates are stated for $\norm{(\hat{m^\varphi})^a\psi}^2$ as these expressions are required for the proof of our main theorem. 
By Lemma~\ref{lem:equivalence:q:m}, they easily translate to bounds on the corresponding quantities $\norm{q_1\cdots q_a\psi}^2$ and $\lr{\xip,\Np^a\xip}$.
The proofs of Proposition~\ref{thm:alpha} and Corollary~\ref{cor:alpha} are postponed to Section~\ref{subsec:proof:thm:alpha}.

\begin{prop}\label{thm:alpha}
Let $\beta\in[0,\frac{1}{\D})$, assume \emph{A1} and \emph{A2} and let $\psi\in L^2_\mathrm{sym}(\R^{\D N})$. Let $s\in\R$, $\ps\in H^k(\R^\D)$ for $k=\lceil\tfrac{\D}{2}\rceil$, and let $\pt$ be the solution of~\eqref{NLS} with initial datum $\ps$.
Then it holds for $t\in\big[s,s+\Tex\big)$ and $j\in\{1\mydots N\}$ that
\lemit{
\item \label{thm:alpha:1}
	for any $b\in\N_0$,
	\begin{eqnarray*}
	\left\|\left(\hat{m^\pt}\right)^j U(t,s)\psi\right\|^2
	&\ls & \cjts\sum\limits_{n=0}^j N^{n(-1+\D\beta)}\left\|\left(\hat{m^\ps}\right)^{j-n}\psi\right	
	\|^2\\
	&&+\,2^b\cbts\sum\limits_{n=0}^b N^{n(-1+\D\beta)+\D\beta b}\left\|\left(\hat{m^\ps}\right)^{b-n}	
	\psi\right\|^2,\end{eqnarray*}
\item \label{thm:alpha:2}
$$\left\|\left(\hat{m^\pt}\right)^j \Ubog(t,s)\psi\right\|^2 \ls 
\cjts\sum\limits_{n=0}^j N^{n(-1+\D\beta)}\left\|\left(\hat{m^\ps}\right)^{j-n}\psi\right\|^2\,,$$
}
where $\cjts:= j!\, 3^{j(j+1)}\e^{9^j \int_s^t\norm{\varphi(s_1)}^2_{H^k(\R^\D)}\d s_1  }$.
\end{prop}

Under the additional assumption A3 on the initial data, this implies that at any time $t$ and for sufficiently large $N$, the first $A$ moments of the number of excitations remain sub-leading:
\begin{cor}\label{cor:alpha}
Assume \emph{A1} -- \emph{A2} and \emph{A3} with $\gamma\in(0,1]$ and $A\in\{1\mydots N\}$. 
Let $\psi(t)$, $\psibog(t)$ and $\pt$ denote the solutions of~\eqref{SE}, \eqref{psibog} and~\eqref{NLS} with initial data $\psi_0$ and $\varphi_0$ from \emph{A3}.
Let 
$$\xipZ=\UNpz\psi_0\,,\quad\xipt=\UNpt\psi(t)=\UNpt U(t,0)\psi_0\,,\quad \xitpt=\UNpt\psibog(t)=\UNpt \Ubog(t,0)\psi_0\,.$$
Then for  $t\in\big[0,\Tex\big)$, sufficiently large $N$ and $a\in\{0\mydots A\}$, it holds that
\lemit{
\item for the time evolution $U(t,s)$ that
	\begin{equation*}
	\norm{(\hat{m^\pt})^a\psi(t)}^2\ls  \cat\begin{cases}
	N^{-a(1-d\beta)} & \text{ for } \beta\in[0,\tfrac{1}{2d})\,, \;  \gamma\in[1-d\beta,1]\,,\\[10pt]
    N^{-\gamma a} & \text{ for }  \beta\in[0,\tfrac{1}{d})\,, \;  \gamma\in(d\beta,1-d\beta]\,,
 \end{cases}\end{equation*}	
 or, equivalently, that
 	\begin{equation*}
	\lr{\xipt,\Npt^a\xipt}_\FNpt\ls  \cat\begin{cases}
	N^{d \beta a} & \text{ for } \beta\in[0,\tfrac{1}{2d})\,, \;  \gamma\in[1-d\beta, 1]\,,\\[10pt]
    N^{(1-\gamma) a} & \text{ for }  \beta\in[0,\tfrac{1}{d})\,, \;  \gamma\in(d\beta, 1-d\beta]\,,
 \end{cases}\end{equation*}	
\item for the time evolution $\Ubog(t,0)$ and $\beta\in[0,\tfrac{1}{d})$ that
\begin{equation*}
 \norm{(\hat{m^\pt})^a\Ubog(t,0)\psi_0}^2\ls  \cat\begin{cases}
 	N^{-a(1-d\beta)}  & \text{ for }\, \gamma\in[1-d\beta, 1]\,,\\[10pt]
  	N^{-\gamma a}     & \text{ for }\,  \gamma\in(0, 1-d\beta]\,,
  \end{cases}
\end{equation*}
or, equivalently, that
\begin{equation*}
\lr{\xitpt,\Npt^a\xitpt}_\FNpt\ls \cat \begin{cases}
	N^{d \beta a} & \gamma\in[1-d\beta, 1]\,,\\[10pt]
	N^{(1-\gamma)a} & \gamma\in(0,1-d\beta]\,
\end{cases}
\end{equation*}}
with $\cat:=C_a^{t,0}$ and where we estimated $a,\fcaz,\fcaz''\ls1$ for the sake of readability.

\end{cor}
At the threshold $\gamma=1-\D\beta$, the leading order terms in the sums in Proposition~\ref{thm:alpha} change, hence we obtain two different estimates. The additional restrictions on $\beta$ and $\gamma$ in part (a) stem from the second sum in Proposition~\ref{thm:alpha:1}.
Only if either $\beta<\frac{1}{2\D}$ or $\gamma>\D\beta$, it is possible to choose $b$ sufficiently large that the first sum dominates for large $N$.
For $\beta=0$, both time evolutions preserve the property A3\,$^{\prime\prime}$ exactly, i.e., with the same power $\gamma$ of $N$, up to a constant growing rapidly in $t$ and $a$.
For $\beta>0$, the conservation is exact only for small $\gamma$, whereas one looses some power of $N$ for larger $\gamma$. 
Further, note that for the range $\gamma\in(0,\D\beta]$, we do not obtain a non-trivial estimate for the excitations $\xipt$ in $U(t,0)\psi_0$.

\subsection{Higher order corrections to the norm approximation}\label{subsec:main:corrections}
Based on the estimates obtained in Proposition~\ref{thm:alpha}, our main result establishes corrections of any order to the norm approximations~\eqref{eqn:norm:approx} and~\eqref{eqn:David:norm}:\ under assumption A3 on the initial data, we construct a sequence $\{\psia\}_{a\in\mathbb{N}}\subset L^2(\mathbb{R}^{\D N})$
such that 
$$
\norm{\psi(t)-\psia(t)}^2\leq C(t) N^{-a\delta(\beta,\gamma)}
$$
for some $\delta(\beta,\gamma)>0$, which may depend on $\beta$ as well as on the parameter $\gamma$ from assumption A3. 
For reasons given below, our analysis is restricted to the scaling regime $\beta\in[0,\frac{1}{4\D})$.\\

As explained in the introduction, it is well known that the actual time evolution $\psi(t)$ is close to the evolution $\psibog(t)$ from~\eqref{psibog} in norm.
Hence, the first element of the approximating sequence $\{\psia\}_{a\in\N}$ is determined by
$$\psio(t):=\Ubog(t,0)\psi_0\,.$$
Using Duhamel's formula, the difference between $U(t,s)\psi$ and $\Ubog(t,s)\psi$ can be expressed as
\begin{eqnarray}\label{eqn:Duhamel1}
U(t,s)\psi
=\Ubog(t,s)\psi-\i\int_s^tU(t,r)\left(\Cpr+\Qpr\right)\Ubog(r,s)\psi\d r
\end{eqnarray}
for any $\psi\in L^2(\R^{\D N})$. Consequently,
\begin{eqnarray}
\norm{\psi(t)-\psio(t)}&=&\left\|-\i\int_0^t U(t,s)\left(\Cps+\Qps\right)\Ubog(s,0)\psi_0\d s\right\|\nonumber\\
&\leq&\int_0^t\norm{\Cps\Ubog(s,0)\psi_0}\d s+\int_0^t\norm{\Qps\Ubog(s,0)\psi_0}\d s\label{eqn:psi1:1}
\end{eqnarray}
by the triangle inequality and as a consequence of the unitarity of $U(t,s)$.
The leading order contribution in~\eqref{eqn:psi1:1} is the term containing $\Cps$
because the cubic interaction terms are larger than the quartic ones in the following sense:
\begin{lem}\label{lem:QC}
Let $\psi\in L^2_\mathrm{sym}(\R^{\D N})$ and denote by $\pt$ the solution of~\eqref{NLS} with initial datum $\pz\in H^k(\R^\D)$, $k=\lceil\tfrac{\D}{2}\rceil$. Then for any $j\in\N_0$ and $t\in\big[0,\Tex\big)$,
\lemit{
	\item $\norm{\left(\hat{m^\pt}\right)^j\Qpt\psi}^2\ls N^{2+2\D\beta}\norm{\big(\hat{m^\pt}\big)^{4+j}\psi}^2$,
	\item $\norm{\left(\hat{m^\pt}\right)^j\Cpt\psi}^2\ls 4^{j} \norm{\pt}^2_{H^k(\R^\D)}N^{2+\D\beta}\norm{\big(\hat{m^\pt}\big)^{3+j}\psi}^2$.
}
\end{lem}
The proof of this lemma is postponed to Section~\ref{subsec:proof:thm:corrections}.
For $j=0$, it gives a bound on the cubic and quartic terms; the more general statement $j\geq0$ is included for later convenience.

When applying Lemma~\ref{lem:QC} to~\eqref{eqn:psi1:1}, we obtain expressions of the form $\norm{(\hat{m^\ps})^j\Ubog(s,0)\psi_0}^2$.
To be able to use assumption A3 on the initial data, we need to interchange, in a sense, the order of $\Ubog(s,0)$ and $(\hat{m^\ps})^j$.
This is where Proposition~\ref{thm:alpha} comes into play: from part \ref{thm:alpha:2}, it follows
for sufficiently large $N$ that
\begin{eqnarray}
\norm{\Cps\Ubog(s,0)\psi_0}^2
&\overset{\text{\ref{lem:QC}}}{\ls}&N^{2+\D\beta}\norm{(\hat{m^\ps})^3\Ubog(s,0)\psi_0}^2\nonumber\\
&\overset{\text{\ref{thm:alpha:2}}}{\ls}&C(s) \, N^{2+\D\beta}\sum_{n=0}^3 N^{n(-1+\D\beta)}\norm{(\hat{m^\pz})^{3-n}\psi_0}^2\nonumber\\
&\overset{\text{A3}}{\ls}&C(s)\, N^{2+\D\beta}\sum_{n=0}^3  N^{n(-1+\D\beta+\gamma)-3\gamma}\,.\nonumber
\end{eqnarray}
To enhance readability, we do not keep track of the different constants  $C(t)$ for now, but we specify it in more detail in Theorem \ref{thm:corrections}.
As in Corollary~\ref{cor:alpha}, the size of $\gamma$ determines the leading order term in the sum: for $\gamma\geq 1-\D\beta$, the dominant contribution issues from $n=3$, whereas otherwise the addend corresponding to $n=0$ is of leading order. Consequently, 
\begin{equation}\label{eqn:psi1:3}
\norm{\Cps\Ubog(s,0)\psi_0}^2\ls C(s)
\begin{cases}
	N^{-1+4\D\beta}&\text{ for }\;\gamma\in[1-\D\beta,1]\,,\\[5pt]
	N^{2+\D\beta-3\gamma} &\text{ for }\;\gamma\in\big(\tfrac{2+\D\beta}{3},1-\D\beta\big]\,.
\end{cases}
\end{equation}
To ensure that~\eqref{eqn:psi1:3} converges to zero as $N\to\infty$, we restricted the range of parameters $\gamma$ admitted by assumption A3 to $\gamma\in(\tfrac{2+\D\beta}{3},1]$. 
Besides, in the first case, the bound is only small for $\beta<\frac{1}{4\D}$, and the second case is anyway only possible for $\beta<\tfrac{1}{4\D}$. 
Hence, we can only cover the parameter regime $\beta\in[0,\tfrac{1}{4\D})$. 
Analogously to~\eqref{eqn:psi1:3}, we also obtain
\begin{equation}\label{eqn:psi2:2}
\norm{\Qps\Ubog(s,0)\psi_0}^2\ls C(s)
	\begin{cases}
	N^{-2+6\D\beta}&\text{ for }\;\gamma\in[1-\D\beta,1]\,,\\[5pt]
	N^{2+2\D\beta-4\gamma}&\text{ for }\;\gamma\in\big(\tfrac{2+\D\beta}{3},1-\D\beta\big]\,.
	\end{cases}
\end{equation}
Note that $\beta<\tfrac{1}{4\D}$ implies that $-2+6\D\beta<-1+4\D\beta$, and besides, it follows from $\gamma>\tfrac{2+3\D}{3}$ and $\beta<\tfrac{1}{4\D}$ that 
 $2+2\D\beta-4\gamma<2+\D\beta-3\gamma$.
Consequently, the contribution with $\Cps$ dominates in~\eqref{eqn:psi1:1} for sufficiently large $N$, which leads to the estimate
\begin{equation}\label{eqn:psi1:2}
\norm{\psi(t)-\psio(t)}^2\ls C(t)N^{-\delta(\beta,\gamma)}
\end{equation}
with 
\begin{equation}\label{delta}
\delta(\beta,\gamma):=\begin{cases}
	1-4\D\beta &\text{ for }\;\gamma\in[1-\D\beta, 1]\,,\\[3pt]
	-2-\D\beta+3\gamma&\text{ for }\;\gamma\in\big(\tfrac{2+\D\beta}{3},1-\D\beta\big]\,.
	\end{cases}
\end{equation}
This yields~\eqref{eqn:aim} for $n=1$.\\

To construct the second element $\psit(t)$ of the approximating sequence, we need to extract from~\eqref{eqn:Duhamel1} the relevant contributions such that  
$\norm{\psi(t)-\psit(t)}^2\leq C(t)N^{-2\delta(\beta,\gamma)}$.
As a consequence of Lemma~\ref{lem:QC}, we define
\begin{eqnarray*}
\psit(t)&:=&\Ubog(t,0)\psi_0-\i\int_0^t\d s\,\Ubog(t,s)\Cps\Ubog(s,0)\psi_0\,,
\end{eqnarray*}
which equals the leading order contribution in~\eqref{eqn:Duhamel1} but with the true time evolution $U(t,s)$ replaced by $\Ubog(t,s)$.
Put differently, the leading order contribution is cancelled but for the difference between $U(t,s)$ and $\Ubog(t,s)$.
Since this difference is evaluated on $\Cps\Ubog(s,0)\psi_0$, which is small in norm, this is an improvement compared to the first order approximation $\psio(t)$. 
To verify this, let us compute the difference between $\psi(t)$ and $\psit(t)$. Using Duhamel's formula twice, we obtain
\begin{eqnarray*}
\psi(t)-\psit(t)
&=&-\i\int_0^t\left(U(t,s)-\Ubog(t,s)\right)\Cps\Ubog(s,0)\psi_0\d s\\
&&-\i\int_0^tU(t,s)\Qps\Ubog(s,0)\psi_0\d s\\
&=&-\int_0^t\d s_1\int_{s_1}^t\d s_2\,U(t,s_2)\left(\Cpst+\Qpst\right)\Ubog(s_2,s_1)\Cpso\Ubog(s_1,0)\psi_0\\
&&-\i\int_0^tU(t,s)\Qps\Ubog(s,0)\psi_0\d s.
\end{eqnarray*}
Due to the unitarity of $U(t,s)$, we obtain with the triangle inequality
\begin{eqnarray}
\norm{\psi(t)-\psit(t)}
&\leq&\int_0^t\d s_1\int_{s_2}^t\d s_2\norm{\Cpst\Ubog(s_2,s_1)\Cpso\Ubog(s_1,0)\psi_0}\nonumber\\
&&+\int_0^t\d s_1\int_{s_1}^t\d s_2\norm{\Qpst\Ubog(s_2,s_1)\Cpso\Ubog(s_1,0)\psi_0}\nonumber\\
&&+\int_0^t\d s\norm{\Qps\Ubog(s,0)\psi_0}\,.\label{eqn:psi-psi2}
\end{eqnarray}
The leading order term in~\eqref{eqn:psi-psi2} can be estimated as
\begin{eqnarray}
&&\hspace{-2cm}\norm{\Cpst\Ubog(s_2,s_1)\Cpso\Ubog(s_1,0)\psi_0}^2\nonumber\\
&\overset{\text{\ref{lem:QC},\ref{thm:alpha:2}}}{\ls} &N^{2+\D\beta}C(s_1,s_2)\,\sum_{n=0}^3N^{n(-1+\D\beta)}\norm{(\hat{m^{\varphi(s_1)}})^{3-n}\Cpso\Ubog(s_1,0)\psi_0}^2\nonumber\\
&\overset{\text{\ref{lem:QC},\ref{thm:alpha:2}}}{\ls}& N^{4+2\D\beta}C(s_1,s_2)\,\sum_{n=0}^3\sum_{l=0}^{6-n}\, N^{(n+l)(-1+\D\beta)}\norm{(\hat{m^\pz})^{6-n-l}\psi_0}^2\nonumber\\
&\overset{\text{A3}}{\ls}& N^{4+2\D\beta}C(s_1,s_2)\,\sum_{n=0}^3\sum_{l=0}^{6-n} N^{(n+l)(-1+\D\beta+\gamma)-6\gamma}\,.\nonumber
\end{eqnarray}
As before, considering the two ranges of $\gamma$ separately yields for sufficiently large $N$
\begin{equation*} 
\norm{\Cpst\Ubog(s_2,s_1)\Cpso\Ubog(s_1,0)\psi_0}^2 \ls C(s_1,s_2)\, N^{-2\delta(\beta,\gamma)}
\end{equation*}
with $\delta(\beta,\gamma)$ from~\eqref{delta}.
Analogously, the second term can be estimated as
\begin{eqnarray*} 
\norm{\Qpst\Ubog(s_2,s_1)\Cpso\Ubog(s_1,0)\psi_0}^2\ls \;C(s_1,s_2)
	\begin{cases}
	N^{-3+10\D\beta}&\text{ for }\;\gamma\in[1-\D\beta, 1]\,,\\[5pt]
	N^{4+3\D\beta-7\gamma}&\text{ for }\;\gamma\in\big(\tfrac{2+\D\beta}{3},1-\D\beta\big]\,,
	\end{cases}
\end{eqnarray*}
and the third term was already treated in~\eqref{eqn:psi2:2}.
Combining all bounds, we obtain
\begin{eqnarray*}
\norm{\psi(t)-\psit(t)}^2&\ls&C(t)N^{-2\delta(\beta,\gamma)}\,,
\end{eqnarray*}
which yields~\eqref{eqn:aim} for $a=2$.\\

Iterating  Duhamel's formula $(a-1)$ times, we construct $\psia(t)$ as an expansion with $a-1$ terms, where the last term contains the true time evolution $U(t,s)$ and all others exclusively contain $\Ubog(t,s)$. Consequently, to construct $\psith(t)$, we iterate~\eqref{eqn:Duhamel1} once more, which yields
\begin{equation*}\begin{split}
&\left(U(t,0)-\Ubog(t,0)\right)\psi\\
&\qquad=-\i\int_0^t\d s\,\Ubog(t,s)\left(\Cps+\Qps\right)\Ubog(s,0)\psi\\
&\qquad\quad\, -\int_0^t\d s_1\int_{s_1}^t\d s_2\,U(t,s_2)\left(\mathcal{C}^{\varphi(s_2)}+\mathcal{Q}^{\varphi(s_2)}\right)\Ubog(s_2,s_1)\left(\mathcal{C}^{\varphi(s_1)}+\mathcal{Q}^{\varphi(s_1)}\right)\Ubog(s_1,0)\psi\,.
\end{split}\end{equation*}
The leading order contributions issue from the first integral and from the expression with two cubic interaction terms. Analogously to above, they determine the next element $\psith$ of the sequence $\{\psia\}_{a\in\N}$ as
\begin{eqnarray*}
\psith(t)&:=&\Ubog(t,0)\psi-\i\int_0^t\d s\,\Ubog(t,s)\left(\Cps+\Qps\right)\Ubog(s,0)\psi_0\\
&&-\int_0^t\d s_1\int_{s_1}^t\d s_2\,\Ubog(t,s_2)\,\mathcal{C}^{\varphi(s_2)}\Ubog(s_2,s_1)\,\mathcal{C}^{\varphi(s_1)}\Ubog(s_1,0)\psi_0\,,
\end{eqnarray*}
and similar calculations as before yield $\norm{\psi(t)-\psith(t)}^2\ls C(t)N^{-3\delta(\beta,\gamma)}$.
Continuing the iteration of~\eqref{eqn:Duhamel1}, we obtain for any $a\geq 1$ and $s_0=0$ the expansion
\begin{eqnarray}
\psi(t)
	&=&\sum\limits_{n=0}^{a-1}(-\i)^{n}\int\limits_0^t\d s_1\int\limits_{s_1}^t\d s_2\mycdots\hspace{-4pt}\int\limits_{s_{n-1}}^t\d s_{n}\,
	\Ubog(t,s_{n})\left(\mathcal{C}^{\varphi(s_{n})}+\mathcal{Q}^{\varphi(s_{n})}\right)\Ubog(s_{n},s_{n-1})\mycdots\times\nonumber\\
	&&\hspace{6.5cm}\times\Ubog(s_2,s_1)\left(\mathcal{C}^{\varphi(s_1)}+\mathcal{Q}^{\varphi(s_1)}\right)\Ubog(s_1,0)\psi_0\nonumber\\
	&&+(-\i)^{a}\int\limits_0^t\d s_1\int\limits_{s_1}^t\d s_2\mycdots\int\limits_{s_{a-1}}^t\d s_{a}\,
	U(t,s_{a})\left(\mathcal{C}^{\varphi(s_{a})}+\mathcal{Q}^{\varphi(s_{a})}\right)\Ubog(s_{a},s_{a-1})\mycdots\times\nonumber\\
	&&\hspace{6.5cm}\times\Ubog(s_2,s_1)\left(\mathcal{C}^{\varphi(s_1)}+\mathcal{Q}^{\varphi(s_1)}\right)\Ubog(s_1,0)\psi_0\nonumber\\
&=&\sum\limits_{n=0}^{a-1}\prod\limits_{\nu=1}^n\left(-\i\int_{s_{\nu-1}}^t\d s_\nu \right)\Ubog(t,s_n)
	\prod\limits_{\l=0}^{n-1}\left(\left(\mathcal{C}^{\varphi(s_{n-\l})}+\mathcal{Q}^{\varphi(s_{n-\l})}\right)\Ubog(s_{n-\l},s_{n-\l-1})\right)\psi_0\nonumber\\
&&+\prod\limits_{\nu=1}^{a}\left(-\i\int_{s_{\nu-1}}^t\d s_\nu \right)U(t,s_{a})
	\prod\limits_{\l=0}^{a-1}\left(\left(\mathcal{C}^{\varphi(s_{a-\l})}+\mathcal{Q}^{\varphi(s_{a-\l})}\right)\Ubog(s_{a-\l},s_{a-\l-1})\right)\psi_0\, .\qquad
	\label{eqn:Duhamel:k}
\end{eqnarray}
All products are to be understood as ordered, i.e.\ $\prod_{\l=0}^L P_\l:=P_0P_1\cdots P_L$ for $L\in\N$ and any expressions $P_\l$.
Extracting the leading contributions in each order, we construct the  sequence $\{\psia(t)\}_{a\in\N}$ as follows:
\begin{definition}\label{def:psik}
Let $I_1^\pt:=\Cpt$ and $I_2^\pt:=\Qpt$.
Define the set 
$$\mathcal{S}^{(k)}_n:=\left\{(j_1\mydots j_n):\; j_\l\in\{1,2\} \text{ for } \l=1\mydots n \text{ and } \sum_{\l=1}^n j_\l=k\right\},$$
i.e., the set of $n$-tuples with elements in $\{1,2\}$ such that the elements of each tuple add  to $k$.
Define for $n\in\N$ and $n\leq k\leq 2n$
\begin{eqnarray*}
\Tnk
&:=&\sum_{(j_1\mydots j_n)\in\mathcal{S}^{(k)}_n}(-\i)^n
	\prod\limits_{\nu=1}^n\left(\;\int\limits_{s_{\nu-1}}^t\d s_\nu\right)
	\Ubog(t,s_n)\prod\limits_{\l=0}^{n-1}\left(I_{j_{n-\l}}^{\varphi(s_{n-\l})}\Ubog(s_{n-\l},s_{n-\l-1})\right)\psi_0
\\
&=&(-\i)^n \int\limits_0^t\d s_1\int\limits_{s_1}^t\d s_2\mycdots\int\limits_{s_{n-1}}^t\d s_n\, 
	\Ubog(t,s_n)\times\\
	&&\times\sum_{(j_1\mydots j_n)\in\mathcal{S}^{(k)}_n}
	\left(I_{j_n}^{\varphi(s_n)}	\Ubog(s_n,s_{n-1})I_{j_{n-1}}^{\varphi(s_{n-1})}\mycdots
	\Ubog(s_2,s_1)I_{j_1}^{\varphi(s_1)}\right)
	\Ubog(s_1,0)\psi_0\,,
\end{eqnarray*}
where  $s_0:=0$.
As above, the products are ordered.
For $n=k=0$, let $\Tzz:=\Ubog(t,0)\psi_0$, and $\Tnk:=0$ for $k< n$ and $k>2n$.
Hence, $\Tnk$ is an $n$-dimensional integral where the integrand contains all possible combinations of $I^{\varphi(s_l)}_{j_l}$ such that $\sum_{l=1}^n j_l=k$.
\\
Finally, the elements of the sequence $\{\psia\}_{a\in\N}$ are defined as
\begin{eqnarray*}
\psia(t)&:=&\sum\limits_{k=0}^{a-1}\sum_{n=\lceil\frac{k}{2}\rceil}^k \Tnk\,
\;=\; \sum\limits_{n=0}^{a-1}\;\sum\limits_{k=n}^{\min\{2n,a-1\}}\Tnk.
\end{eqnarray*}
\end{definition}

\begin{thm}\label{thm:corrections}
Let $\beta\in[0,\frac{1}{4\D})$ and assume \emph{A1} -- \emph{A3} with $A\in\{1\mydots N\}$ and $\gamma\in(\tfrac{2+\D\beta}{3},1]$.
Let $\psi(t)$ and $\pt$ denote the solutions of~\eqref{SE} and~\eqref{NLS} with initial data $\psi_0$ and $\varphi_0$ from \emph{A3}, respectively,
and let $\psi^{(a)}_\varphi(t)$ be defined as in Definition~\ref{def:psik}. Then for sufficiently large $N$, $t\in\big[0,\Tex\big)$ and $a\in\{1\mydots\lfloor\tfrac{A}{6}\rfloor\}$, there exists a constant $c(a)$ such that
\begin{equation*}
\norm{\psi(t)-\psia(t)}^2\ls \e^{c(a)\int\limits_0^t\norm{\varphi(s)}^2_{H^k(\R^\D)}\d s}\; N^{-a\delta(\beta,\gamma)},
\end{equation*}
where 
\begin{equation*}
\delta(\beta,\gamma)=\begin{cases}
	1-4\D\beta &\text{ for }\;\gamma\in[1-\D\beta,1]\,,\\[3pt]
	3\gamma-2-\D\beta&\text{ for }\;\gamma\in\big(\tfrac{2+\D\beta}{3},1-\D\beta\big]\,.
	\end{cases}
\end{equation*}
\end{thm}

Hence, given any desired precision of the approximation, there exists some $a\in\N$ such that the corresponding function $\psia(t)$ approximates the actual $N$-body dynamics $\psi(t)$ to this order for large $N$. 
To compute $\psia(t)$, an $a$-dependent  number of steps is required, as well as the knowledge of the first quantised Bogoliubov time evolution. 
We cover initial states where the first $A$ moments of the number of excitations are sub-leading, where $A$ depends on $a$ but is independent of $N$.

\section{Proofs}\label{sec:proofs}
\subsection{Preliminaries}\label{subsec:preliminaries}
\begin{lem}\label{lem:pt}
Let $\pz\in H^k(\R^\D)$ for $k=\lceil\tfrac{\D}{2}\rceil$, $t\in\big[0,\Tex\big)$ and $\pt$ the solution of~\eqref{NLS} with inital datum $\pz$.
\lemit{
	\item 	\label{lem:pt:2}	
			Let $f:\R^\D\times\R^\D\rightarrow\R$ be a measurable function such that $|f(z_j,z_k)|\leq F(z_k-z_j)$ almost everywhere for some 
			$F:\R^\D\rightarrow\R $. Then 
			$$\onorm{\ppt_1f(x_1,x_2)}\ls\norm{\pt}_{H^k(\R^\D)}\norm{F}_{L^2(\R^\D)}.$$
	\item 	\label{lem:pt:3}
			Let $f:\mathbb{N}_0\rightarrow\R^+_0$. Then
			$\Ppt_k,\;\hat{f^\pt}\in \mathcal{C}^1\big(\R,\mathcal{L}\left(L^2(\R^{\D N})\right)\big)$ for $0\leq k\leq N$
			and
			$$\tfrac{\d}{\d t}\hat{f^\pt}=\i\Big[\hat{f^\pt},\sum\limits_{j=1}^N \hpt_j(t)\Big],$$
			where $\hpt_j(t)$ denotes the one-particle operator $\hpt(t)$ from \eqref{NLS} acting on the $j$\textsuperscript{th} coordinate.		
}
\end{lem}
\begin{proof}
For part (a), see, e.g.,~\cite[Lemma 4.1]{pickl2015} and note that $\norm{\pt}_{L^\infty(\R^\D)}\ls\norm{\pt}_{H^k(\R^\D)}$ by the Sobolev embedding theorem.
Part (b) can be shown as in the proof of~\cite[Lemma 6.2]{pickl2015}.
\end{proof}

\begin{lem}\label{lem:fqq}
Let $\psi\in L^2_\mathrm{sym}(\R^{\D N})$, $\varphi\in L^2(\R^d)$ and $f:\mathbb{N}_0\rightarrow\mathbb{R}_0^+$.
\lemit{
	\item 	\label{lem:fqq:1}
			$\left(\hat{n^\varphi}\right)^2=\frac{1}{N}\sum\limits_{j=1}^N \qp_j\,.$
	\item	\label{lem:fqq:2}
			Let $a\in\{1\mydots N\}$. Then for $j\in\{0\mydots a\}$,
			$$\norm{\qp_1\cdots \qp_a\psi}^2 \leq \norm{\qp_1\cdots \qp_j \left(\hat{n^\varphi}\right)^{a-j}\psi}^2\,.$$
	\item	\label{lem:fqq:3}
			In particular, this implies
			$$\left\|\hat{f^\varphi}\qp_1\psi\right\|^2\leq \left\|\hat{f^\varphi}\hat{n^\varphi}\psi\right\|
			^2, \qquad 
			\left\|\hat{f^\varphi}\qp_1\qp_2\psi\right\|^2\leq\left\|\hat{f^\varphi}\left(\hat{n^\varphi}\right)^2\psi\right\|^2\,. $$
}
\end{lem}
\begin{proof}
For simplicity, let us drop all superscripts $\varphi$.
Part (a) is shown e.g.\ in~\cite[Lemma 4.1]{pickl2015}. For part (b), observe that for any $1\leq j\leq N$,
\begin{eqnarray*}
\norm{q_1\cdots q_j\psi}^2
&=&\tfrac{j-1}{N}\lr{\psi,q_1\cdots q_{j}\psi}+\tfrac{N-j+1}{N}\lr{\psi,q_1\cdots q_j\psi}\\
&\leq&\tfrac{1}{N}\lr{\psi,q_1\cdots q_{j-1}\left(j-1+(N-j+1) q_j\right)\psi}\\
&=&\lr{\psi,q_1\cdots q_{j-1}\left(\tfrac{1}{N}\sum\limits_{l=1}^N q_l\right)\psi}
=\norm{q_1\cdots q_{j-1}\hat{n}\psi}^2
\end{eqnarray*}
by part (a). Since $\hat{n}\psi$ is again symmetric, the statement follows by iteration.
\end{proof}

\begin{lem}
Denote by $T_{ij}$ an operator acting non-trivially only on coordinates $i$ and $j$.
\lemit{
	\item	\label{lem:commutators:2}
			Let $\varphi\in L^2(\R^\D)$, let $f,g:\mathbb{N}_0\rightarrow\mathbb{R}_0^+$ be any weights and $i,j\in\{1\mydots N\}$. 
			Let  $Q^\varphi_0:=\pp_i\pp_j$, $Q^\varphi_1\in\{\pp_i\qp_j,\qp_i\pp_j\}$ and $Q^\varphi_2:=\qp_i\qp_j$. 
			Then, for $\mu,\nu\in\{0,1,2\}$,	
			$$ Q^\varphi_\mu\,\hat{f^\varphi}\,T_{ij}\,Q^\varphi_\nu=Q^\varphi_\mu \, T_{ij}\hat{f^\varphi}_{\mu-\nu}\,Q^\varphi_\nu\,.$$
	\item 	\label{lem:GammaLambda}
			Let $\Gamma,\Lambda\in L^2(\R^{\D N})$ be symmetric under the exchange of coordinates in a subset $\mathcal{M}\subseteq\{1\mydots N\}$ such that $j\notin\mathcal{M}$ and $k,l\in\mathcal{M}$. Then
		\begin{equation*}
		|\lr{\Gamma,T_{j,k}\Lambda}|\leq\norm{\Gamma}
		\Big(|\lr{T_{j,k}\Lambda,T_{j,l}\Lambda}|+|\mathcal{M}|^{-1}\norm{T_{j,k}\Lambda}^2\Big)^\frac12.
		\end{equation*}	
}
\end{lem}
\begin{proof}
\cite[Lemma 4.1]{pickl2015} and \cite[Lemma 4.7]{NLS}.
\end{proof}

\noindent\emph{Proof of Lemma~\ref{lem:equivalence:q:m}.}
Let us for simplicity drop all superscripts $\varphi$. 
First, observe that 
\begin{equation}\label{eqn:n:m}\begin{cases}
n(k)^{2a}=\left(\tfrac{k}{N}\right)^a\leq \left(\tfrac{k+1}{N}\right)^a = m(k)^{2a}& \text{ for } k\geq 0,\\[2pt]
m(k)^{2a} \leq  \left(\tfrac{2k}{N}\right)^a =2^a n(k)^{2a}&\text{ for }k\geq 1\,,
\end{cases}\end{equation}
hence 
\begin{equation}\label{eqn:n:m:op}
\hat{n}^{2a}\leq \hat{m}^{2a}\leq 2^a\hat{n}^{2a}+N^{-a}
\end{equation}
in the sense of operators.
The first part of (a) follows from Lemma~\ref{lem:fqq:2} and the first line in~\eqref{eqn:n:m}.
For the second part, Lemma~\ref{lem:fqq:1} implies
\begin{eqnarray*}
\norm{\hat{n}^a\psi}^2 &=&\lr{\psi,\left(\frac{1}{N}\sum\limits_{j=1}^N q_j\right)^a\psi} = N^{-a}\lr{\psi,\sum\limits_{a_1+\mycdots+a_N=a} \binom{a}{a_1\mydots a_N}\,q_1^{a_1}\cdots q_N^{a_N}\psi}
\end{eqnarray*}
for $a_1\mydots a_N\in\{0\mydots a\}$. Due to the symmetry of $\psi$, since there are $\binom{a-1}{j-1}$ possibilities to write $a$ as the sum of $j$ positive integers and with $\binom{a}{a_1\mydots a_N}\leq a!$, this yields
\begin{eqnarray*}
\norm{\hat{n}^a\psi}^2 &=&\frac{a!}{N^a}\sum\limits_{j=1}^a \binom{N}{j}\binom{a-1}{j-1}\norm{q_1\cdots q_j\psi}^2\,.
\end{eqnarray*}
Further, note that 
\begin{equation} 
\max\limits_{j=\{1\mydots a-1\}}\binom{a-1}{j-1}=\binom{a-1}{\lceil\frac{a-1}{2}\rceil}
=\frac{(a-1)!}{\lceil\frac{a-1}{2}\rceil!\lfloor\frac{a-1}{2}\rfloor!}\leq 2^{a-1}\,,
\end{equation}
and $\tbinom{N}{j}\leq N^j$, hence
$$\norm{\hat{m}^a\psi}^2\leq N^{-a}\left(1+2^{2a-1}a! \sum\limits_{j=1}^a \tbinom{N}{j}\norm{q_1\mycdots q_j\psi}^2\right)\,.$$
Part (b) follows from~\eqref{eqn:N} and~\eqref{eqn:n:m:op}.
\qed

\subsection{Proof of Proposition~\ref{thm:alpha}}\label{subsec:proof:thm:alpha}
\emph{Proof of Proposition~\ref{thm:alpha}.}
The proof of this proposition is essentially an adaptation of the proof of~\cite[Corollary 4.2]{petrat2017}.
We begin with part (a).
Let $\psi\in L^2(\R^{\D N})$ symmetric, $s\in\R$ and $f:\N_0\to\R_0^+$ some weight function. Define 
\begin{equation}\label{eqn:thm1:0}
\alpha_{\psi,\varphi,s}(f;t):=\lr{U(t,s)\psi,\hat{f^\pt}\,U(t,s)\psi}.
\end{equation}
and
\begin{equation}\label{Z^beta}
Z^\beta_{ij}:=\left(\vbij-\vbarpt(x_i)-\vbarpt(x_j)+2\mpt\right).
\end{equation}
Let us for the moment abbreviate $U(t,s)\psi=:\psi_t$.
By Lemma~\ref{lem:pt:3},
\begin{eqnarray}
\tfrac{\d}{\d t}\alpha_{\psi,\varphi,s}(f;t)
&=&\i\lr{\psi_t,\left[\Hb(t)-\sum\limits_{j=1}^N\hpt_j(t),\,\hat{f^\pt}\right]\psi_t}\nonumber\\
&=&\i\tfrac{N}{2}\lr{\psi_t,\left[Z^\beta_{12},\,\hat{f^\pt}\right]\psi_t}\nonumber\\
&=&2N\Im\lr{\psi_t,\left(\hat{f^\pt}-\hat{f^\pt_{-1}}\right)\qpt_1\ppt_2Z^\beta_{12}\ppt_1\ppt_2\psi_t}\label{eqn:thm1:1}\\
&&+N\Im\lr{\psi_t,\left(\hat{f^\pt}-\hat{f^\pt_{-2}}\right)^\frac12\qpt_1\qpt_2\vb_{12}\ppt_1\ppt_2\left(\hat{f^\pt_2}-\hat{f^\pt}\right)^\frac12\psi_t}\label{eqn:thm1:2}\\
&&+2N\Im\lr{\psi_t,\left(\hat{f^\pt}-\hat{f^\pt_{-1}}\right)^\frac12\qpt_1\qpt_2Z^\beta_{12}\ppt_1\qpt_2\left(\hat{f^\pt_1}-\hat{f^\pt}\right)^\frac12 \psi_t},\qquad\label{eqn:thm1:3}
\end{eqnarray}
where we have inserted $\mathbbm{1}=(\ppt_1+\qpt_1)(\ppt_2+\qpt_2)$ on both sides of the commutator and used Lemma~\ref{lem:commutators:2}.
Since $\qpt_1\ppt_2Z^\beta_{12}\ppt_1\ppt_2=0$, we conclude that \eqref{eqn:thm1:1} equals zero.
From now on, we will for simplicity drop the superscripts $\pt$.
Let 
\begin{equation}\begin{split}\label{Lf}
{L}_f:=\Bigg\{\sum\limits_{k=2}^N(f(k)-f(k&-2))\Ppt_k, \;
				   \sum\limits_{k=1}^N\left(f(k)-f(k-1)\right)\Ppt_k, \;\\
				   &\sum\limits_{k=0}^{N-2}\left(f(k+2)-f(k)\right)\Ppt_k, \;
				   \sum\limits_{k=0}^{N-1}\left(f(k+1)-f(k)\right)\Ppt_k \bigg\}\,.
\end{split}\end{equation}
Since, for example, 
$ \left(\hat{f}-\hat{f_{-2}}\right)^\frac12 q_1q_2
=\left(\sum\limits_{k=2}^N(f(k)-f(k-2))\Ppt_k\right)^\frac12q_1q_2$, this yields
\begin{eqnarray}\label{aux:1}
\tfrac{\d}{\d t}\alpha_{\psi,\varphi,s}(f;t)
&\ls& \max\limits_{\hat{l}\in L_f} 
	\left\{N\left|\lr{\psi_t,\hat{l}^\frac12 q_1q_2\vbot p_1p_2\hat{l}^\frac12\psi_t}\right|
	+N\left|\lr{\psi_t,\hat{l}^\frac12 q_1q_2Z^\beta_{12}p_1q_2\hat{l}^\frac12\psi_t}\right|\right\}\,.\qquad
\end{eqnarray}
By Lemmas~\ref{lem:pt} and~\ref{lem:fqq} and since $\norm{\vb}_{L^2(\R^\D)}^2\ls N^{\D\beta}$, the first term in~\eqref{aux:1} leads to
\begin{eqnarray}
&&\hspace{-1cm} N\left|\lr{\psi_t,\hat{l}^\frac12 q_1q_2\vbot p_1p_2\hat{l}^\frac12\psi_t}\right|\nonumber\\
&\ls&N\norm{\hat{l}^\frac12 q_1\psi_t}\left(
	\lr{q_2\vbot p_2\,\hat{l}^\frac12p_1\psi_t,q_3\vboth p_3\,\hat{l}^\frac12 p_1\psi_t}
	+N^{-1}\norm{q_2\vbot p_2 p_1\,\hat{l}^\frac12\psi_t}^2\right)^\frac12\nonumber\\
&\ls&N\norm{\hat{l}^\frac12\,q_1\psi_t}
\bigg(\norm{\hat{l}^\frac12q_3\psi_t} \onorm{p_1p_2\vbot\vboth p_3p_1}\norm{\hat{l}^\frac12q_2 \psi_t}
+N^{-1}\onorm{\vbot p_2}^2\norm{\hat{l}^\frac12\psi_t}^2\bigg)^\frac12\nonumber\\
&\ls& N\lr{\psi_t,\hat{l}\,\hat{n}^2\psi_t}^\frac12\left(\lr{\psi_t,\hat{l}\,\hat{n}^2\psi_t}+N^{-1+\D\beta}\lr{\psi_t,\hat{l}\psi_t}\right)^\frac12\norm{\pt}_{H^k(\R^\D)}^2\,,\label{aux:2}
\end{eqnarray}
To obtain the estimate in the last line, note first that
$$ \onorm{p_1p_2\vboth\vbot p_1p_3}=\onorm{p_1\vboth p_2p_3\vbot p_1}=\onorm{p_1\vbot p_2}^2 \,.$$
Now we decompose $\vb=\vb_+-\vb_-$ into its positive and negative part such that $\vb_\pm\geq0$, hence $\vb_\pm(x)=\sqrt{\vb_\pm(x)}\sqrt{\vb_\pm(x)}$, 
which leads to
\begin{eqnarray*}
\onorm{p_1\vbot p_2}&=&\onorm{p_1(\vb_+-\vb_-)_{12} p_2}\\
&\leq& \onorm{p_1\sqrt{(\vb_+)_{12}}\sqrt{(\vb_+)_{12}}p_2}+\onorm{p_1\sqrt{(\vb_-)_{12}}\sqrt{(\vb_-)_{12}}p_2}\\
&\ls& \norm{\pt}_{H^k(\R^\D)}^2\left(\norm{\vb_+}_{L^1(\R^d)}+\norm{\vb_-}_{L^1(\R^d)}\right)\\
&=& \norm{\pt}_{H^k(\R^\D)}^2\norm{\vb}_{L^1(\R^\D)}\;\ls\; \norm{\pt}_{H^k(\R^\D)}^2
\end{eqnarray*}
by Lemma~\ref{lem:pt}.
The second term in~\eqref{aux:1} can be estimated as
\begin{eqnarray}
N\left|\lr{\psi_t,\hat{l}^\frac12 q_1q_2Z^\beta_{12}p_1q_2\hat{l}^\frac12\psi_t}\right|
&\ls&  N\norm{\hat{l}^\frac12 q_1q_2\psi_t}\norm{\hat{l}^\frac12\hat{n}\psi_t}\onorm{Z^\beta_{12}p_1}\nonumber\\
&\ls& N^{1+\frac{\D\beta}{2}}\lr{\psi_t,\hat{l}\,\hat{n}^4\psi_t}^\frac12\lr{\psi_t,\hat{l}\,\hat{n}^2\psi_t}^\frac12\norm{\pt}_{H^k(\R^\D)}\,.\label{aux:3}
\end{eqnarray}
Now we choose for $f$ the family of weight functions $\wl^j:k\mapsto\left(\wl(k)\right)^j$ given by
\begin{equation}\label{eqn:wl}
\wl(k):=
\begin{cases}
	\displaystyle \frac{k+1}{N^\lambda} & 0\leq k \leq N^\lambda -1,\\[8pt]
	\;\; 1 & \mathrm{else}
\end{cases}
\end{equation}
for some $0<\lambda\leq1-\D\beta$ and $j\in\{0\mydots N\}$. 
The set corresponding to $L_f$ from~\eqref{Lf} is called $L_{w^j_\lambda}$.
To bound the operators in $L_{w^j_\lambda}$, note that for any $a,b\in\mathbb{N}_0$, $a>b$, 
\begin{eqnarray*}
(k+a)^j-(k+b)^j
&=&\tbinom{j}{j-1}k^{j-1}(a-b)+\tbinom{j}{j-2}k^{j-2}(a^2-b^2)+...+(a^j-b^j)\\
&\leq& ja^j\left(\tbinom{j-1}{j-1}k^{j-1}+\tbinom{j-1}{j-2}k^{j-2}+...+\tbinom{j-1}{1}k+1\right)\\
&=&j a^j(k+1)^{j-1},
\end{eqnarray*}
where we have used in the second line that for every $1\leq m\leq j-1$,
$$\tbinom{j}{m}=\tfrac{j(j-1)!}{(j-m)((j-1)-m)!m!}=\tfrac{j}{j-m}\tbinom{j-1}{m}\leq j\tbinom{j-1}{m},$$
and that $a^j\geq a^\l-b^\l$ for any $1\leq\l\leq j$ and $j\geq 1$ (the statement is trivial for $j=0$).
Since $\wl(k)\leq \tfrac{k+1}{N^\lambda}$ for all $k$, especially also if $k>N^{\lambda}-1$, we conclude that
\begin{eqnarray*}
(\wl(k))^j-(\wl(k-1))^j&\hspace{-0.4cm}\leq\tfrac{(k+1)^j-k^j}{N^{\lambda j}}\leq j\tfrac{(k+1)^{j-1}}{N^{\lambda j}} = j\tfrac{\wl(k)^{j-1}}{N^\lambda}
\;\,\qquad\quad &\text{for } 1\leq k\leq N^\lambda-1,\\
(\wl(k+1))^j-(\wl(k))^j&\leq\tfrac{(k+2)^j-(k+1)^j}{N^{\lambda j}}\leq j2^j\tfrac{(k+1)^{j-1}}{N^{\lambda j}}=j2^{j}\tfrac{\wl(k)^{j-1}}{N^\lambda}\quad &\text{for } 0\leq k\leq N^\lambda-1,\\
(\wl(k+2))^j-(\wl(k))^j&\leq\tfrac{(k+3)^j-(k+1)^j}{N^{\lambda j}}\leq j3^j\tfrac{(k+1)^{j-1}}{N^{\lambda j}}=j3^{j}\tfrac{\wl(k)^{j-1}}{N^\lambda}\quad &\text{for } 0\leq k\leq N^\lambda-1.
\end{eqnarray*}
Besides, one computes analogously to above that $(k+1)^j-(k-1)^j\leq 2j(k+1)^{j-1}$, hence
\begin{eqnarray*}
(\wl(k))^j-(\wl(k-2))^j&\leq&\tfrac{(k+1)^j-k^j}{N^{\lambda j}}\leq 2j\tfrac{(k+1)^{j-1}}{N^{\lambda j}}=2j\tfrac{\wl(k)^{j-1}}{N^\lambda} \hspace{1.5cm} \text{for } 2\leq k\leq N^\lambda-1\,.
\end{eqnarray*}
Finally, $\wl(k)=1$ for $k>N^\lambda-1$, hence the above estimates imply
\begin{eqnarray*}
(\wl(k))^j-(\wl(k-1))^j&\leq j\tfrac{(k+1)^{j-1}}{N^{\lambda j}} \leq j2^{j-1}N^{-\lambda}= j2^{j-1}\tfrac{\wl(k)^{j-1}}{N^\lambda}\quad&\text{for }N^\lambda-1<k\leq N^\lambda,\\
(\wl(k))^j-(\wl(k-2))^j&\hspace{-0.5cm}\leq2j\tfrac{(k+1)^{j-1}}{N^{\lambda j}} \leq  j2^{j}N^{-\lambda}=j2^{j}\tfrac{\wl(k)^{j-1}}{N^\lambda}
\quad &\text{for } N^\lambda-1\leq k\leq N^\lambda.
\end{eqnarray*}
For all other values of $k$, the differences yield zero.
Thus, every element of $L_{\wl^j}$ can be bounded, in the sense of operators, by the operator corresponding to the weight function
\begin{equation}\label{lj}
\llj(k)= \begin{cases}
	\displaystyle j3^j\frac{\wl(k)^{j-1}}{N^\lambda}& 0\leq k\leq N^\lambda,\\[5pt]
	\; 0 & \mathrm{else}.\end{cases}
\end{equation}
Besides, since $\llj(k)=0$ for $k> N^\lambda+1$, one obtains
\begin{eqnarray}
\llj(k)n^2(k)&\leq& j3^j N^{-1} \wl^j(k),\label{eqn:thm1:7}\\
\llj(k) n^4(k)&\leq& j3^j\wl^j(k)\tfrac{k}{N^2}\leq j3^j\wl^j(k)\tfrac{N^\lambda+1}{N^2}
\ls j3^j N^{-2+\lambda}\wl^j(k).\label{eqn:thm1:8}
\end{eqnarray}
Inserting~\eqref{lj} to~\eqref{eqn:thm1:8} into~\eqref{aux:2} and~\eqref{aux:3} and using that $\lambda\leq 1-\D \beta$ implies $N^\frac{\D \beta+\lambda-1}{2}\leq 1$, we conclude that
\begin{equation}
\tfrac{\d}{\d t}\alpha_{\psi,\varphi,s}(\wl^j;t)\;\ls\; j3^j\norm{\pt}_{H^k(\R^\D)}^2\left(\alpha_{\psi,\varphi,s}(\wl^j;t)+ N^{\D\beta-\lambda}\alpha_{\psi,\varphi,s}(\wl^{j-1};t)\right)\,.
\end{equation}
Now we apply Grönwall's inequality, for now on using the abbreviations 
$\alpha_{\psi,\varphi,s}(\wl^j;t)=:\alpha_j(t)$
and $I_t:=\int_s^t\norm{\varphi(s_1)}_{H^k(\R^\D)}^2\d s_1$.
This yields
\begin{eqnarray*}
\alpha_j(t) &\ls&\e^{j3^jI_t}\left(\alpha_j(s)+j3^j N^{\D\beta-\lambda}\int_s^t\norm{\varphi(s_1)}_{H^k(\R^\D)}^2\alpha_{j-1}(s_1)\d s_1\right)\\
&\leq&\e^{j3^jI_t}\alpha_j(s)
+j3^j\e^{j(3^j+3^{j-1}) I_t}I_tN^{\D\beta-\lambda}\alpha_{j-1}(s)\\
&&+j(j-1)3^{j+(j-1)}\e^{j(3^j+3^{j-1}) I_t}I_t^2N^{2(\D\beta-\lambda)}\int_s^{t}\norm{\varphi(s_1)}_{H^k(\R^\D)}^2\alpha_{j-2}(s_1)\d s_1\\
&\ls&\e^{j3^jI_t}\alpha_j(s)\\
&&+j3^j\e^{j(3^j+3^{j-1}) I_t}I_tN^{\D\beta-\lambda}\alpha_{j-1}(s)\\
&&+j(j-1)3^{j+(j-1)}\e^{j(3^j+3^{j-1}+3^{j-2})I_t}I_t^2N^{2(\D\beta-\lambda)}\alpha_{j-2}(s)\\
&&+j(j-1)(j-2)3^{j+(j-1)+(j-2)}\e^{j(3^j+3^{j-1}+3^{j-2})I_t}I_t^2N^{3(\D\beta-\lambda)}\times\\
&&\hspace{5cm}\times\int_s^{t}\norm{\varphi(s_1)}_{H^k(\R^\D)}^2\alpha_{j-3}(s_1)\d s_1\\
&\ls& ...\\
&\ls&\sum\limits_{n=0}^{j}\tfrac{j!}{(j-n)!}3^\frac{n(2j+1-n)}{2}\e^{2j3^jI_t}I_t^n N^{n(\D\beta-\lambda)}\alpha_{j-n}(s)\,,
\end{eqnarray*}
where we have used that all integrands are non-negative and thus the upper boundary of all integrals could be replaced by $t$.
Written explicitly, this gives
\begin{equation}
\alpha_{\psi,\varphi,s}(\wl^j;t)\ls \cjts\,\sum\limits_{n=0}^{j}N^{n(\D\beta-\lambda)}\alpha_{\psi,\varphi,s}(\wl^{j-n};s)
=\cjts\,\sum\limits_{n=0}^{j}N^{n(\D\beta-\lambda)}\lr{\psi,\hat{\wl}^{j-n}\psi},\label{eqn:thm1:9}
\end{equation}
with
$$\cjts:= j!\, 3^{j(j+1)}\e^{9^j \int_s^t\norm{\varphi(s_1)}^2_{H^k(\R^\D)}\d s_1  }\,,$$
where we have estimated $I_t^j\e^{2j3^jI_t}<\e^{9^jI_t}$.
To relate this  estimate to $\norm{\hat{m}^j\psi}^2$, observe that for $0\leq k\leq N$,
$$\wl^j(k)\leq\left(\tfrac{k+1}{N^\lambda}\right)^j=\left(\tfrac{k+1}{N}\right)^j N^{j(1-\lambda)}=m^{2j}(k) N^{j(1-\lambda)},$$
and
$$m^{2j}(k)\;=\;\left(\tfrac{k+1}{N}\right)^j\leq
	\begin{cases}
		\left(\tfrac{k+1}{N^\lambda}\right)^jN^{-j(1-\lambda)}\;=\;\wl^j(k)N^{-j(1-\lambda)} & \text{ for }\, 0\leq k\leq N^\lambda-1,\\[5pt]
		2^j\;=\;2^j\wl^b(k)\; \text{ for any }b\in\N & \text{ for }\, N^\lambda-1\leq k\leq N.
	\end{cases}
$$
Consequently, $m^{2j}(k)\leq N^{-j(1-\lambda)}\wl^j(k)+\wl^b(k)$, and we conclude
\begin{eqnarray*}
	\alpha_{\psi,\varphi,s}(\wl^j;t)&=&\lr{\psi_t,\hat{\wl}^j\psi_t}
	\;\leq\; N^{j(1-\lambda)}\lr{\psi_t,\hat{m}^{2j}\psi_t}\;=\;N^{j(1-\lambda)}\norm{\hat{m}^j\psi_t}^2,\\
	\norm{\hat{m}^j\psi_t}^2&=&\lr{\psi_t,\hat{m}^{2j}\psi_t} \;\leq\; N^{-j(1-\lambda)}\alpha_{\psi,\varphi,s}(\wl^j;t)+2^j\alpha_{\psi,\varphi,s}(\wl^b;t)
\end{eqnarray*}
for any $b\in\N$. Inserting these estimates into~\eqref{eqn:thm1:9} yields
\begin{eqnarray*}
\norm{\hat{m}^jU(t,s)\psi}^2 &\ls &
\cjts\,\sum\limits_{n=0}^j N^{n(-1+\D\beta)}\norm{\hat{m}^{j-n}\psi}^2
+2^j\cbts\,\sum\limits_{n=0}^b N^{n(-1+\D\beta)+b(1-\lambda)}\norm{\hat{m}^{b-n}\psi}^2.
\end{eqnarray*}
To minimise the second term, we choose the maximal $\lambda=1-\D\beta$, which concludes the proof of part (a).\\

The proof of part (b) is much simpler since we now consider the time evolution $\Ubog(t,s)$. The term corresponding to~\eqref{eqn:thm1:3} vanishes, which implies that we may directly consider the weights $m^{2j}(k)$ instead of taking the detour via $\wl^j(k)$. 
Analogously to~\eqref{eqn:thm1:0}, we define 
\begin{equation*}
\tilde{\alpha}_{\psi,\varphi,s}(f;t):=\lr{\Ubog(t,s)\psi,\hat{f^\pt}\,\Ubog(t,s)\psi}\,.
\end{equation*} 
We will now abbreviate $\Ubog(t,s)\psi=:\tilde{\psi}_t$. In this notation,
\begin{eqnarray*}
\tfrac{\d}{\d t}\tilde{\alpha}_{\psi,\varphi,s}(f;t)
&=&\i\lr{\tilde\psi_t,\left[\Hpt(t)-\sum\limits_{j=1}^N\hpt_j(t),\,\hat{f^\pt}\right]\tilde\psi_t}\\
&=&\i\tfrac{N}{2}\lr{\tilde\psi_t,\left[\ppt_1\qpt_2\vbot\qpt_1\ppt_2+\mathrm{h.c.}\,,\,\hat{f^\pt}\right]\tilde\psi_t}\\
&&+\i\tfrac{N}{2}\lr{\tilde\psi_t,\left[\ppt_1\ppt_2\vbot\qpt_1\qpt_2+\mathrm{h.c.}\,,\,\hat{f^\pt}\right]\tilde\psi_t}\\
&=&-N\Im\lr{\tilde\psi_t,\qpt_1\qpt_2\left(\hat{f^\pt}-\hat{f^\pt_{-2}}\right)^\frac12\vbot\ppt_1\ppt_2\left(\hat{f^\pt_2}-\hat{f^\pt}\right)^\frac12\tilde\psi_t}.
\end{eqnarray*}
We now evaluate this expression for the weight $m^{2j}(k)$, i.e.
$$\tilde\alpha_{\psi,\varphi,s}(m^{2j};t)=\left\|\left(\hat{m^\pt}\right)^j\tilde\psi_t\right\|^2.$$
This corresponds to $\wl^j(k)$ with the choice $\lambda=1$ in~\eqref{eqn:wl}.  Consequently, we define $\lj(k):=j3^jN^{-1}m^{2(j-1)}(k)$
analogously to~\eqref{lj} and conclude that $m^{2j}(k)-m^{2j}(k-2)\leq \lj(k)$ and  $m^{2j}(k+2)-m^{2j}(k)\leq \lj(k)$.
Analogously to the estimate of the first term in~\eqref{aux:1} and using the relation~\eqref{eqn:thm1:7} for $\lambda=1$, we obtain
$$\tfrac{\d}{\d t}\big\|\left(\hat{m^\pt}\right)^j\tilde{\psi}\big\|^2
\ls j3^j\norm{\pt}^2_{H^k(\R^\D)}\left(\big\|{\left(\hat{m^\pt}\right)^j\tilde{\psi}}\big\|^2
+N^{-1+\D\beta}\big\|{\left(\hat{m^\pt}\right)^{j-1}\tilde{\psi}}\big\|^2\right).$$
The same Grönwall argument which led to~\eqref{eqn:thm1:9} concludes the proof.
\qed\\

\noindent\emph{Proof of Corollary~\ref{cor:alpha}.}
From Proposition~\ref{thm:alpha:1} and the assumptions on the initial data, we conclude that for every $b\in\N$ and sufficiently large $N$,
\begin{eqnarray*}
\left\|\left(\hat{m^\pt}\right)^a\psi(t)\right\|^2
&\ls&\cat\,\sum\limits_{n=0}^a \mathfrak{C}_{a-n}\,N^{n(-1+\D\beta+\gamma)-\gamma a}
+2^b\cbt\sum\limits_{n=0}^b \mathfrak{C}_{b-n}\,N^{n(-1+\D\beta+\gamma)-b(\gamma-\D\beta)}\,.
\end{eqnarray*}
If $\gamma\geq 1-\D\beta$, the leading order terms in both sums are the ones with maximal $n$, hence
$$\left\|\left(\hat{m^\pt}\right)^a\psi(t)\right\|^2 \ls (a+1)\cat\,N^{a(-1+\D\beta)}
+(b+1)\cbt\,N^{b(-1+2\D\beta)}.$$
If one chooses $b>a\frac{1-\D\beta}{1-2\D\beta}$ for fixed $\beta<\frac{1}{2\D}$, the second term is for sufficiently large $N$ dominated by the first one.
For $\gamma<1-\D\beta$, the leading order terms are those with $n=0$, hence
$$\left\|\left(\hat{m^\pt}\right)^a\psi(t)\right\|^2 \ls (a+1)\cat\,\fcaz\,N^{-\gamma a}
+(b+1)2^b\cbt\,\fcbz\,N^{-b(\gamma-\D\beta)}\,,$$
which yields a non-trivial bound only for $\gamma>\D\beta$.
Part (b) follows analogously from part (b) of Proposition~\ref{thm:alpha} without the restrictions on $\beta$ and $\gamma$ that are due to the second sum.
\qed

\subsection{Proof of Theorem~\ref{thm:corrections}}\label{subsec:proof:thm:corrections}
\noindent\emph{Proof of Lemma~\ref{lem:QC}.}
We use the abbreviation
$Z^\beta_{ij}=\vbij-\vbarpt(x_i)-\vbarpt(x_j)+2\mpt$ as in~\eqref{Z^beta},
and drop all superscripts $\pt$ in $\ppt$, $\qpt$  and $\hat{m^\pt}$ for simplicity. 
By Lemma~\ref{lem:commutators:2}, $\Qpt\hat{m}^a=\hat{m}^a\Qpt$, hence
\begin{eqnarray*}
\norm{\hat{m}^a\Qpt\psi}^2
&=&\tfrac{1}{(N-1)^2}\sum\limits_{i<j}\sum\limits_{k<l}\lr{\hat{m}^a\psi, q_iq_j Z^\beta_{ij}q_iq_jq_kq_lZ^\beta_{kl}q_kq_l\hat{m}^a\psi}\\
&=&\tfrac{N}{2(N-1)}\lr{\hat{m}^a\psi,q_1q_2Z^\beta_{12}q_1q_2Z^\beta_{12}q_1q_2\hat{m}^a\psi}\\
&&+\tfrac{N(N-2)}{N-1}\lr{\hat{m}^a\psi,q_1q_2Z^\beta_{12}q_1q_2q_3Z^\beta_{13}q_1q_3\hat{m}^a\psi}\\
&&+\tfrac{N(N-2)(N-3)}{4(N-1)}\lr{\hat{m}^a\psi,q_1q_2Z^\beta_{12}q_1q_2q_3q_4Z^\beta_{34}q_3q_4\hat{m}^a\psi}\\
&\ls&N^{2\D\beta}\left(\norm{q_1q_2\hat{m}^a\psi}^2+N\norm{q_1q_2q_3\hat{m}^a\psi}^2+N^2\norm{q_1q_2q_3q_4\hat{m}^a\psi}^2\right),
\end{eqnarray*}
where we have used that 
$\norm{Z^\beta_{ij}}_{L^\infty(\R^\D)}\ls N^{\D\beta}$
by Young's inequality.
Now observe that 
$$\tbinom{N}{2}\norm{q_1q_2\hat{m}^a\psi}^2=\sum\limits_{i<j}\lr{\hat{m}^a\psi,q_iq_j\hat{m}^a\psi}<\sum\limits_{i,j}\lr{\hat{m}^a\psi,q_iq_j\hat{m}^a\psi}
<\sum\limits_{i,j,k,l}\lr{\hat{m}^a\psi,q_iq_jq_kq_l\hat{m}^a\psi},$$
hence
\begin{eqnarray*}
\norm{q_1q_2\hat{m}^a\psi}^2&\ls& N^{-2}\sum\limits_{i,j,k,l}\lr{\hat{m}^a\psi,q_iq_jq_kq_l\hat{m}^a\psi}\;=\;N^2\lr{\hat{m}^a\psi,\bigg(\tfrac{1}{N}\sum\limits_{j=1}^Nq_j\bigg)^4\hat{m}^a\psi}\\
&=&N^{2}\lr{\hat{m}^a\psi,\hat{n}^8\hat{m}^a\psi}\;<\;N^{2}\norm{\hat{m}^{4+a}\psi}^2,
\end{eqnarray*}
by~\eqref{eqn:n:m}, and analogously
\begin{eqnarray*}
\norm{q_1q_2q_3\hat{m}^a\psi}^2&=&\tbinom{N}{3}^{-1}\sum\limits_{i<j<k}\lr{\hat{m}^a\psi,q_iq_jq_k\hat{m}^a\psi}\\
	&\ls& N^{-3}\sum\limits_{i,j,k,l}\lr{\hat{m}^a\psi,q_iq_jq_kq_l\hat{m}^a\psi}\;\ls\; N\norm{\hat{m}^{4+a}\psi}^2 ,\\
\norm{q_1q_2q_3q_4\hat{m}^a\psi}^2&=&\tbinom{N}{4}^{-1}\sum\limits_{i<j<k<l}\lr{\hat{m}^a\psi,q_iq_jq_kq_l\hat{m}^a\psi}\\
	&\ls& N^{-4}\sum\limits_{i,j,k,l}\lr{\hat{m}^a\psi,q_iq_jq_kq_l\hat{m}^a\psi}
	\;\ls\;  \norm{\hat{m}^{4+a}\psi}^2.
\end{eqnarray*}
This implies part (a).
For part (b), note that by Lemma~\ref{lem:commutators:2},
\begin{equation*}
\hat{m}^a\Cpt=\tfrac{1}{N-1}\sum\limits_{i<j}\left(q_iq_jZ^\beta_{ij}(q_ip_j+p_iq_j)\right)\hat{m}^a_1
+\tfrac{1}{N-1}\sum\limits_{i<j}\left((p_iq_j+q_ip_j)Z^\beta_{ij}q_iq_j\right)\hat{m}^a_{-1}.
\end{equation*}
Consequently,
\begin{eqnarray*}
&&\hspace{-1cm}\norm{\hat{m}^a\Cpt\psi}^2\\
&=&\tfrac{1}{(N-1)^2}\sum\limits_{i<j}\sum\limits_{k<l}\Big(
\lr{\hat{m}_1^a\psi,(q_ip_j+p_iq_j)Z^\beta_{ij}q_iq_jq_kq_lZ^\beta_{kl}(p_kq_l+q_kp_l)\hat{m}^a_{1}\psi}\\
&&\qquad\qquad+\lr{\hat{m}_1^a\psi,(q_ip_j+p_iq_j)Z^\beta_{ij}q_iq_j(p_kq_l+q_kp_l)Z^\beta_{kl}q_lq_k\hat{m}^a_{-1}\psi}\\
&&\qquad\qquad+\lr{\hat{m}^a_{-1}\psi,q_iq_jZ^\beta_{ij}(p_iq_j+q_ip_j)q_kq_lZ^\beta_{kl}(p_kq_l+q_kp_l)\hat{m}^a_{1}\psi}\\
&&\qquad\qquad+\lr{\hat{m}^a_{-1}\psi,q_iq_jZ^\beta_{ij}(p_iq_j+q_ip_j)(p_kq_l+q_kp_l)Z^\beta_{kl}q_kq_l\hat{m}^a_{-1}\psi}\Big)\\
&\ls& N^{\D\beta}\left(\norm{q_1\hat{m}^a_1\psi}^2+\norm{q_1q_2\hat{m}^a_{-1}\psi}^2\right)\norm{\pt}^2_{H^k(\R^\D)}\\
&&+N^{1+\D\beta}
	\big(\norm{q_1q_2\hat{m}^a_1\psi}^2+\norm{q_1\hat{m}^a_{1}\psi}\norm{q_1q_2q_3\hat{m}^a_{-1}\psi}+\norm{q_1q_2\hat{m}^a_{-1}\psi}^2\big)\norm{\pt}^2_{H^k(\R^\D)}\\
&&+N^{2+\D\beta}\big(\norm{q_1q_2q_3\hat{m}^a_1\psi}^2+\norm{q_1q_2q_3\hat{m}^a_{-1}\psi}^2+\norm{q_1q_2\hat{m}^a_1\psi}\norm{q_1q_2q_3q_4\hat{m}^a_{-1}\psi}\big)\norm{\pt}^2_{H^k(\R^\D)}
\end{eqnarray*}
similarly to the estimate of $\norm{\hat{m}^a\Qpt\psi}$.
The last inequality follows because by Lemma~\ref{lem:pt:2}, $\onorm{p_1Z^\beta_{12}}^2\ls N^{\D\beta}\norm{\pt}^2_{H^k(\R^\D)}$ due to Young's inequality and since $\norm{\vb}_{L^2(\R^\D)}^2\ls N^{\D\beta}$. 
Further, note that
\begin{eqnarray*}
\hat{m}^{2a}_1&=&\left(\sum\limits_{k=0}^{N-1}m(k+1)P_k\right)^{2a}
\;=\;\left(\sum\limits_{k=0}^{N-1}\sqrt{\tfrac{k+2}{N}}P_k\right)^{2a}
\;\leq\; \left(2\sum\limits_{k=0}^{N}\sqrt{\tfrac{k+1}{N}}P_k\right)^{2a}
\;=\; 4^{a}\hat{m}^{2a}\,,\\
\hat{m}^{2a}_{-1}&=&\left(\sum\limits_{k=1}^{N}m(k-1)P_k\right)^{2a}
\;=\;\left(\sum\limits_{k=1}^{N}\sqrt{\tfrac{k}{N}}P_k\right)^{2a}
\;\leq\; \left(\sum\limits_{k=0}^{N}\sqrt{\tfrac{k+1}{N}}P_k\right)^{2a}
\;=\; \hat{m}^{2a}
\end{eqnarray*}
in the sense of operators.
As in the estimate of $\Qpt$, we thus obtain for $\ell\in\{-1,1\}$
\begin{eqnarray*}
\norm{q_1\hat{m}^a_\ell\psi}^2&<&N^{-1}\sum\limits_{i,j,k}\lr{\hat{m}^a_\ell\psi,q_iq_jq_k\hat{m}^a_\ell\psi}= N^2\lr{\hat{n}^3\psi,\hat{m}^{2a}_\ell\,\hat{n}^3\psi}
\leq  2^{2a}N^2\norm{\hat{m}^{a+3}\psi}^2\,,
\end{eqnarray*}
and analogously $\norm{q_1q_2\hat{m}^a_\ell\psi} < 4^{a}N\norm{\hat{m}^{a+3}\psi}^2$ and $\norm{q_1q_2q_3\hat{m}^a_\ell\psi} < 4^{a}\norm{\hat{m}^{a+3}\psi}^2$.
Together, this implies part (b).
\qed\\

\noindent\emph{Proof of Theorem~\ref{thm:corrections}.}
Let $a\in\N_0$ such that $6a\leq A$.
Recall that by Definition~\ref{def:psik},
\begin{eqnarray*}
\psiao(t)&=&\sum\limits_{n=0}^a\sum\limits_{k=n}^{\min\{2n,a\}}\Tnk
\end{eqnarray*}
for any $a\geq 0$, where $\Tnk$ is given by
$$\Tnk=\sum_{(j_1\mydots j_n)\in\mathcal{S}^{(k)}_n}(-\i)^n\prod\limits_{\nu=1}^n\left(\;\int\limits_{s_{\nu-1}}^t\d s_\nu\right)
	\Ubog(t,s_n)\,t_{(j_1\mydots j_n)}^{(k)}\,,
$$
where  
$$t_{(j_1\mydots j_n)}^{(k)}:=\begin{cases}
	0& \text{for }k<n \text{ and } k>2n,\\[4pt]
	\psi_0& \text{for }k=n=0, \\[4pt]
	\prod\limits_{\l=0}^{n-1} \left(I^{\varphi(s_{n-\l})}_{j_{n-\l}}\Ubog(s_{n-\l},s_{n-\l-1})\right)\psi_0 & \text{else},
\end{cases}$$
with $I^\pt_1=\Cpt$ and $I^\pt_2=\Qpt$ and $(j_1\mydots j_n)\in \mathcal{S}_n^{(k)}$.
In this notation, 
$$\prod\limits_{\l=0}^{n-1}\left(\left(\mathcal{C}^{\varphi(s_{n-\l})}+\mathcal{Q}^{\varphi(s_{n-\l})}\right)\Ubog(s_{n-\l},s_{n-\l-1}\right)=
\sum\limits_{k=n}^{2n}\sum\limits_{(j_1\mydots j_n)\in\mathcal{S}_n^{(k)}}t_{(j_1\mydots j_n)}^{(k)}\,,$$
hence the Duhamel expansion~\eqref{eqn:Duhamel:k} of $\psi(t)$ reads
$$\psi(t)
=\sum\limits_{n=0}^{a-1}\sum\limits_{k=n}^{2n}\Tnk
+\sum\limits_{k=a}^{2a}\Ttak\,.
$$
Here, $\Ttnk$ is obtained from $\Tnk$ by replacing the first $\Ubog(t,s_n)$ by  the full time evolution $U(t,s_n)$, i.e., for $n<k<2n$,
$$\Ttnk:=\sum_{(j_1\mydots j_n)\in\mathcal{S}^{(k)}_n}(-\i)^n\prod\limits_{\nu=1}^n\left(\;\int\limits_{s_{\nu-1}}^t\d s_\nu\right)
	U(t,s_n)\prod\limits_{l=0}^{n-1}\left(I_{j_{n-l}}^{\varphi(s_{n-l})}\Ubog(s_{n-l},s_{n-l-1})\right)\psi_0\,.$$
Consequently,
\begin{eqnarray}
\psi(t)-\psiao(t)
&=&\sum\limits_{n=0}^{a-1}\,\sum\limits_{k=\min\{2n,a\}+1}^{2n}\Tnk + \sum\limits_{k=a}^{2a}\Ttak-\sum\limits_{k=a}^{\min\{2a,a\}}\Tak\nonumber\\
&=&\sum\limits_{n=\lceil\frac{a+1}{2}\rceil}^{a-1}\sum\limits_{k=a+1}^{2n}\Tnk + \sum\limits_{k=a+1}^{2a}\Ttak
+\left(\Ttaa-\Taa\right)\label{eqn:proof:thm:corr:1}
\end{eqnarray}
since the first double sum contributes only if $2n\geq a+1$, and in this case $\min\{2n,a\}=a$.
Note that for $k=n$, $j_1=\dots=j_k=1$, hence $\Tkk$ and $\Ttkk$ exclusively contain $\mathcal{C}^{\varphi(s_l)}$.
Using Duhamel's formula, the last expression can thus be expanded as
\begin{eqnarray}
&&\hspace{-1cm}\Ttaa-\Taa\nonumber\\
&=&(-\i)^a\int\limits_0^t\d s_1\mycdots\int\limits_{s_{a-1}}^t\d s_a 
	\left(U(t,s_a)-\Ubog(t,s_a)\right)\mathcal{C}^{\varphi(s_a)}\Ubog(s_a,s_{a-1})\mathcal{C}^{\varphi(s_{a-1})}\mycdots\mathcal{C}^{\varphi(s_{1})}\Ubog(s_1,0)\psi_0	\nonumber\\
&=&(-\i)^{a+1}\int\limits_0^t\d s_1\mycdots\int\limits_{s_a}^t\d s_{a+1}U(t,s_{a+1})\left(\mathcal{C}^{\varphi(s_{a+1})}+\mathcal{Q}^{\varphi(s_{a+1})}\right)\Ubog(s_{a+1},s_a)\mathcal{C}^{\varphi(s_a)}\times\nonumber\\
&&\hspace{9.5cm}\times\mycdots\,\mathcal{C}^{\varphi(s_{1})}\Ubog(s_1,0)\psi_0\nonumber\\
&=&\Ttaoao+(-\i)^{a+1}\int\limits_0^t\d s_1\mycdots\int\limits_{s_a}^t\d s_{a+1}U(t,s_{a+1})\,t^{(a+2)}_{(1,1\mydots 1,2)}\,.\label{eqn:proof:thm:corr:2}
\end{eqnarray}
By unitarity of $U(t,s)$ and $\Ubog(t,s)$,
\begin{eqnarray*}
\norm{\Tnk}&\leq& \sum_{(j_1\mydots j_n)\in\mathcal{S}^{(k)}_n}\int\limits_0^t\d s_1\mycdots\int\limits_0^t\d s_n\norm{t_{(j_1\mydots j_n)}^{(k)}}\,, \\
\norm{\Ttnk}&\leq& \sum_{(j_1\mydots j_n)\in\mathcal{S}^{(k)}_n}\int\limits_0^t\d s_1\mycdots\int\limits_0^t\d s_n\norm{t_{(j_1\mydots j_n)}^{(k)}} \,.
\end{eqnarray*}
With this,~\eqref{eqn:proof:thm:corr:1} and~\eqref{eqn:proof:thm:corr:2} imply for $a=0,1$
\begin{eqnarray}
\norm{\psi(t)-\psio(t)}&=&\bigg\|\Ttoo-\i\int\limits_0^t\d s_1 U(t,s_1)t^{(2)}_{(2)}\bigg\|\;\leq\; 
2\max\limits_{k\in\{1,2\}}\left\{\int\limits_0^t\d s\norm{t^{(k)}_{(k)}}\right\}\,,\label{a=0}\\
\norm{\psi(t)-\psit(t)}&=&\bigg\|\Ttot+\Tttt-\int\limits_0^t\d s_1\int\limits_{s_1}^t\d s_2\,U(t,s_2)t^{(3)}_{(1,2)}\bigg\|\nonumber\\
&\leq& 3\max\limits_{\substack{n\in\{1,2\}\\k\in\{2,3\}}}\left\{\sum\limits_{(j_1\mydots j_n)\in\mathcal{S}_n^{(k)}}\int\limits_0^t\d s_1\mycdots\int\limits_0^t\d s_n\norm{t^{(k)}_{(j_1\mydots j_n)}}\right\}\label{a=1}
\end{eqnarray}
which coincides with~\eqref{eqn:psi1:1} and~\eqref{eqn:psi-psi2}. For $a\geq 2$, we find
\begin{eqnarray}
&&\hspace{-1cm}\norm{\psi(t)-\psiao(t)}\nonumber\\
&<&a^2\max\limits_{\substack{n\in\{\lceil\frac{a+1}{2}\rceil\mydots a-1\}\\k\in\{a+1\mydots 2(a-1)\}}}\norm{\Tnk }
+ a\max\limits_{k\in\{a+1\mydots 2a\}}\norm{\Ttak}
+\norm{\Ttaoao}+\int\limits_0^t\d s_1\mycdots\int\limits_{s_a}^t\d s_{a+1}\left\|t^{(a+2)}_{(1,1\mydots 1,2)}\right\|
\nonumber\\
&\leq&2a^2 \max\limits_{\substack{n\in\{\lceil\frac{a+1}{2}\rceil\mydots a+1\}\\  k\in\{a+1\mydots 2a\}}}\left\{
\sum_{(j_1\mydots j_n)\in\mathcal{S}^{(k)}_n}\int\limits_0^t\d s_1\mycdots\int\limits_{s_{n-1}}^t\d s_{n}\left\|t_{(j_1\mydots j_n)}^{(k)}\right\|\right\}\nonumber\\
&\ls&a^2 \max\limits_{\substack{k\in\{a+1\mydots 2a\}\\n\leq k}}\left\{
\sum_{(j_1\mydots j_n)\in\mathcal{S}^{(k)}_n}\int\limits_0^t\d s_1\mycdots\int\limits_{s_{n-1}}^t\d s_{n}\left\|t_{(j_1\mydots j_n)}^{(k)}\right\|\right\}\label{eqn:proof:thm:corr:3}
\end{eqnarray}
where we used that $a+2\leq 2a$ for $a\geq 2$.
To estimate $\norm{t_{(j_1\mydots j_n)}^{(k)}}^2$ for $a+1\leq k\leq 2a$ and $n\leq k$, note first that Lemma~\ref{lem:QC} and Proposition~\ref{thm:alpha:2}
can be combined into the single statement
\begin{equation}\begin{split}\label{eqn:combi}
&\left\|\left(\hat{m^\pt}\right)^a I^\pt_j \Ubog(t,s) \psi\right\|^2\\
&\qquad\ls 4^{a}\norm{\pt}^2_{H^k(\R^\D)} N^{2+\D\beta j}C_{2+a+j}^{\,t-s}
\sum\limits_{\nu=0}^{2+j+a}N^{\nu(-1+\D\beta)}\left\|\left(\hat{m^\ps}\right)^{2+j+a-\nu}\psi\right\|^2
\end{split}\end{equation}
for $j\in\{1,2\}$ and any $\psi\in L^2_\mathrm{sym}(\R^{\D N})$.
Hence, with $\delta_\mu:=2(n-\mu+1)+(j_n+j_{n-1}+\mycdots+j_\mu) $ and $\eta_\mu:=\prod_{\l=0}^\mu\norm{\varphi(s_{n-\l})}^2_{H^k(\R^\D)}$, we obtain for $n\leq k$
\begin{eqnarray}
&&\hspace{-1cm}\norm{t_{(j_1\mydots j_n)}^{(k)}}^2\nonumber\\
&\ls& N^{2+\D\beta j_n}
	\sum\limits_{\nu_1=0}^{\delta_n} C_{\delta_n}^{\,s_n-s_{n-1}}\eta_0\, N^{\nu_1(-1+\D\beta)}
	\left\|\left(\hat{m^{\varphi(s_{n-1})}}\right)^{\delta_n-\nu_1}  
		\prod\limits_{\l=1}^{n-1} \left(I^{\varphi(s_{n-\l})}_{j_{n-\l}}\Ubog(s_{n-\l},s_{n-\l-1})\right)\psi_0\right\|^2\nonumber\\
&\ls& N^{2\cdot 2+\D\beta (j_n+j_{n-1})}\eta_1
		\sum\limits_{\nu_1=0}^{\delta_n} \sum\limits_{\nu_2=0}^{\delta_{n-1}-\nu_1}
		4^{\delta_n-\nu_1} C_{\delta_n}^{s_n-s_{n-1}}\,C_{\delta_{n-1}-\nu_1}^{s_{n-1}-s_{n-2}}\, N^{(\nu_1+\nu_2)(-1+\D\beta)}\times\nonumber\\
	&&\times\left\|\left(\hat{m^{\varphi(s_{n-2})}}\right)^{\delta_{n-1}-(\nu_1+\nu_2)}  
		\prod\limits_{\l=2}^{n-1} \left(I^{\varphi(s_{n-\l})}_{j_{n-\l}}\Ubog(s_{n-\l},s_{n-\l-1})\right)\psi_0\right\|^2\nonumber\\
&\ls&\dots\nonumber\\
&\ls&  N^{2(\mu+1)+\D\beta (j_n+\mycdots+ j_{n-\mu})}\eta_\mu 
	\sum\limits_{\nu_1=0}^{\delta_n} \mycdots 
	\sum\limits_{\nu_{\mu+1}=0}^{\delta_{n-\mu}-(\nu_1+\mycdots+\nu_{\mu})}C_{\delta_n}^{\,s_n-s_{n-1}} \mycdots C_{\delta_{n-\mu}-(\nu_1+\mycdots+\nu_{\mu})}^{\,s_{n-\mu}-s_{n-\mu-1}}
		\times \nonumber\\ 
	&& \hspace{2cm}\times 4^{\delta_n+\mycdots+\delta_{n+1-\mu}-(\nu_1+\mycdots+(\nu_1+\mycdots+\nu_{\mu})}\,
	 N^{(\nu_1+\mycdots+\nu_{\mu+1})(-1+\D\beta)}\nonumber\\
	&&\times \left\|\left(\hat{m^{\varphi(s_{n-\mu-1})}}\right)^{\delta_{n-\mu}-(\nu_1+\mycdots+\nu_{\mu+1}))}  	
	\prod\limits_{\l=\mu+1}^{n-1} \left(I^{\varphi(s_{n-\l})}_{j_{n-\l}}\Ubog(s_{n-\l},s_{n-\l-1})\right)\psi_0\right\|^2\nonumber\\	
&\ls&\dots\nonumber\\
&\ls&  N^{2n+\D\beta (j_n+\mycdots+ j_1)}\eta_{n-1}
	\sum\limits_{\nu_1=0}^{\delta_n} \mycdots 
	\sum\limits_{\nu_n=0}^{\delta_1-(\nu_1+\mycdots+\nu_{n-1})}
		4^{\delta_n+\mycdots+\delta_{2}-(\nu_1+\mycdots+(\nu_1+\mycdots+\nu_{n-1}))}\times \nonumber\\ 
	&&\times C_{\delta_n}^{\,s_n-s_{n-1}} \mycdots C_{\delta_1-(\nu_1+\mycdots+\nu_{n-1})}^{\,s_{1}}\,
	 N^{(\nu_1+\mycdots+\nu_n)(-1+\D\beta)}\left\|\left(\hat{m^{\varphi_0}}\right)^{\delta_{1}-(\nu_1+\mycdots+\nu_n)}  	
	\psi_0\right\|^2.	\label{eqn:proof:thm:corr:4}
\end{eqnarray}
Since $j_1+\mydots+j_n=k$ and $n\leq k\leq 2a$, we find $\delta_1=2n+k\leq 3k\leq 6a\leq A$, hence assumption A3 yields
$$\norm{(\hat{m^{\varphi_0}})^{\delta_1-(\nu_1+\mycdots+\nu_n)}\psi_0}^2\ls \fC_{\,\delta_1-(\nu_1+\mycdots+\nu_n)}\,N^{-\gamma\delta_1+\gamma(\nu_1+\mycdots+\nu_n)}\,.$$
Let us for the moment focus on the $N$-dependent factors in~\eqref{eqn:proof:thm:corr:4}, thereby neglecting all other contributions to the sum. This yields
$$ N^{2n+\D\beta k-\gamma\delta_1}\sum\limits_{\nu_1=0}^{\delta_n} \mycdots 
	\sum\limits_{\nu_n=0}^{\delta_1-(\nu_1+\mycdots+\nu_{n-1})}
	N^{(\nu_1+\mycdots+\nu_n)(-1+\D\beta+\gamma)}\,.
$$
For $\gamma\geq 1-\D\beta$, the leading order term in the sum $\sum_{\nu_n}$ is the term corresponding to the choice $\nu_n=\delta_1-(\nu_1+\mycdots+\nu_{n-1})=2n+k-(\nu_1+\mycdots+\nu_{n-1})$, which yields the total factor $N^{k(-1+\D\beta)}N^{\D\beta\delta_1}= N^{-k+2\D\beta(n+k)}$.
This factor is maximal for $n=k$.
For $\gamma< 1-\D\beta$, the leading term corresponds to the choice $\nu_1=\dots=\nu_n=0$, which yields $N^{2n(1-\gamma)+k(\D\beta-\gamma)}$. Also here, the maximal contribution issues from $n=k$.
In fact, the leading contributions for both ranges of $\gamma$ can be summarised as $N^{-k\delta(\beta,\gamma)}$, where
$$
\delta(\beta,\gamma)=\begin{cases}
	1-4\D\beta &\text{ for }\;1-\D\beta\leq\gamma\leq 1\,,\\[3pt]
	-2-\D\beta+3\gamma&\text{ for }\;\tfrac{2+\D\beta}{3}<\gamma\leq 1-\D\beta
	\end{cases}
$$ 
as defined in~\eqref{delta}.
Hence, for sufficiently large $N$, the dominating terms is the one with $n=k$, which comes from $t^{(k)}_{(j_1\mydots j_k)}=t^{(k)}_{(1\mydots 1)}$.
\begin{equation*}
\max\limits_{(j_1\mydots j_n)\in\mathcal{S}^{(k)}_n}\left\| t^{(k)}_{(j_1\mydots j_n)}\right\| = \left\| t^{(k)}_{(1\mydots 1)}\right\|\,,
\end{equation*}
and~\eqref{a=0} to~\eqref{eqn:proof:thm:corr:3} can be summarised as
\begin{eqnarray}
\norm{\psi(t)-\psiao(t)}
&\leq&(a+1)^2 \max\limits_{\substack{a+1\leq k\leq\max\{2a,a+2\}}}\left\{
\int\limits_0^t\d s_1\mycdots\int\limits_{s_{k-1}}^t\d s_{k}\left\|t_{(1\mydots 1)}^{(k)}\right\|\right\}\,.\label{eqn:proof:thm:corr:5}
\end{eqnarray}
It remains to evaluate the estimate~\eqref{eqn:proof:thm:corr:4} for $n=k$.
In this case, $j_1=\dots=j_k=1$ and $\delta_\mu=3(k-\mu+1)$. Note also that the constants $C_{a}^{\,t}$ are increasing in $a$ and $t$, hence  
$C_{\delta_{k-\mu}-(\nu_1+\mycdots+\nu_{\mu-1})}^{\,s_{k-\mu}-s_{k-\mu-1}}\leq  C_{3(\mu+1)}^{\,s_{k-\mu}}$.
Further, observe that $\delta_k+\mycdots+\delta_2= \tfrac{3}{2}k(k-1)\leq \tfrac{3}{2}k^2$.
Consequently,
\begin{eqnarray}
\norm{t_{(1\mydots 1)}^{(k)}}^2
&\ls& (1+\fC_{\,3k})\,2^{3k^2} N^{-k\delta(\beta,\gamma)}\prod\limits_{\mu=0}^{k-1}\left((3\mu+1) C^{\,s_{k-\mu}}_{3(\mu+1)}\norm{\varphi(s_\mu)}^2_{H^k(\R^\D)} \right) \,,
\label{eqn:proof:thm:corr:6}
\end{eqnarray}
where we have used that each sum $\sum_{\nu_{\mu}}$ in~\eqref{eqn:proof:thm:corr:4} contains at most
$\delta_{k-\mu+1}+1=3\mu+1$ addends, 
and that the prefactor of the leading order term for $\gamma\geq 1-\D\beta$ is $\fczz=1$, whereas it is $\fC_{\,3k}$ for $\gamma< 1-\D\beta$.  
Consequently, for sufficiently large $N$, the maximum in~\eqref{eqn:proof:thm:corr:5} is attained for $k=a+1$. Inserting the explicit formula $\cjts =j!\,3^{j(j+1)}\e^{9^{j}I_t} $ with $I_t=\int_s^t\norm{\varphi(s_1)}^2_{H^k(\R^\D)}\d s_1$ yields
\begin{eqnarray*}
\norm{\psi(t)-\psia(t)}^2&\ls&
 N^{-a\delta(\beta,\gamma)}\prod\limits_{\nu=1}^a \left(\int_0^t \e^{\frac12 9^{3(\nu+1)} I_{s_\nu}}\norm{\varphi(s_n)}_{H^k(\R^\D)}\d s_\nu\right)^2\\
&\ls&  \e^{a9^{3(a+1)}I_t}I_t^{2a}\,N^{-a\delta(\beta,\gamma)}\;\ls\;
\e^{9^{4a}I_t}\,N^{-a\delta(\beta,\gamma)}
\,,
\end{eqnarray*}
where we have bounded all $a$-dependent, time-independent expressions by a constant $c\ls1$. \qed

\section*{Acknowledgements}
\begin{wrapfigure}{l}{0.088\textwidth}
 \vspace{-15pt}
\includegraphics[scale=0.27]{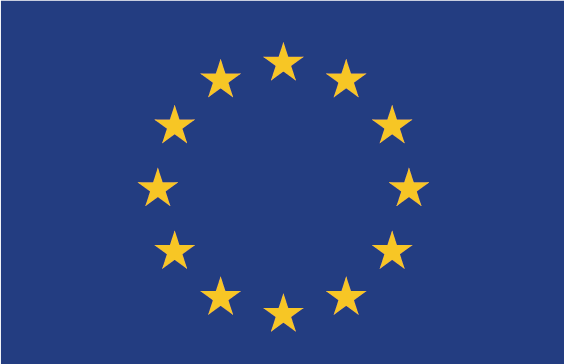}
  \vspace{-11pt}
\end{wrapfigure}
L.B.\ gratefully acknowledges the support by the German Research Foundation (DFG) within the Research Training Group 1838 ``Spectral Theory and Dynamics of Quantum Systems'', and wishes to thank Stefan Teufel, Sören Petrat and Marcello Porta for helpful discussions.  This project has received funding from the European Union’s Horizon 2020 research and innovation programme under the Marie Sk{\textl}odowska-Curie Grant Agreement No.\ 754411.
N.P.\ gratefully acknowledges support from NSF grant DMS-1516228 and DMS-1840314. 
P.P.'s research was funded by DFG Grant no.\ PI 1114/3-1. 
Part of this work was done when N.P. and P.P.\ were visiting CCNU, Wuhan. N.P.\ and P.P.\ thank A.S.\ for his hospitality at CCNU.

\renewcommand{\bibname}{References}
\bibliographystyle{abbrv}

   \bibliography{bib_PhD}
\end{document}